\documentclass[11pt]{article} \usepackage[normal]{optional}
\usepackage{amsmath}
\usepackage{amssymb}
\usepackage{ifpdf}
\usepackage{cite}
\usepackage{pageno}

\usepackage[text={6.5in,9in},centering]{geometry}
\opt{normal}{\usepackage{commands-tam,numbering-tam}}
\opt{submission}{
    \usepackage{commands-tam}
    \numberwithin{equation}{section}
    \numberwithin{theorem}{section}
    \spnewtheorem{cor}[theorem]{Corollary}{\bfseries}{\itshape}
    \spnewtheorem{lem}[theorem]{Lemma}{\bfseries}{\itshape}
    \spnewtheorem{prop}[theorem]{Proposition}{\bfseries}{\itshape}
    \spnewtheorem{obs}[theorem]{Observation}{\bfseries}{\itshape}
}

\renewcommand{\vec}[1]{\mathbf{#1}}

\newcommand{\qedl}{\opt{normal}{}\opt{submission,final}{\qed}}

\usepackage{url}
\usepackage{amsfonts}
\opt{normal}{\usepackage{amsthm}}
\opt{normal}{\usepackage{fullpage}}

\opt{submission,final}{}
\opt{normal}{}

\newcommand{\ignore}[1]{}

\newcommand{\rest}{\upharpoonright}

\usepackage[normalem]{ulem}
\newcommand{\repds}[2]{#2}
\newcommand{\repdd}[2]{#2}

\ifpdf

  \usepackage[pdftex]{epsfig}
  \usepackage[pdfstartview=FitH,pdfpagemode=None,colorlinks=true,citecolor=blue,linkcolor=blue]{hyperref}

\else
    \usepackage[dvips]{epsfig}
\usepackage[dvips,colorlinks=true,citecolor=blue,linkcolor=blue]{hyperref}

\fi

\begin{document}

\title{Deterministic Function Computation with Chemical Reaction Networks\thanks{The first author was supported by the Molecular Programming Project under NSF grant 0832824, the second and third authors were supported by a Computing Innovation Fellowship under NSF grant 1019343. The second author was supported by NSF grants CCF-1219274 and CCF-1162589. The third author was supported by NIGMS Systems Biology Center grant P50 GM081879.}
}

\opt{normal}{
    \author{
        Ho-Lin Chen\thanks{National Taiwan University, Taipei, Taiwan, {\tt holinc@gmail.com}}
        \and
        David Doty\thanks{California Institute of Technology, Pasadena, CA, USA, {\tt ddoty@caltech.edu}}
        \and
        David Soloveichik\thanks{University of California, San Francisco, San Francisco, CA, USA, {\tt david.soloveichik@ucsf.edu}}
    }
    \date{}
}

\opt{submission}{
    \author{Ho-Lin Chen\inst{1}, David Doty\inst{2}, and David Soloveichik\inst{3}}
    \institute{National Taiwan University, Taipei, Taiwan \email{holinc@gmail.com}
    \and
    California Institute of Technology, Pasadena, California, USA \email{ddoty@caltech.edu}
    \and
    University of California, San Francisco, San Francisco, CA, USA \email{david.soloveichik@ucsf.edu}
    }
}
\newcommand{\calN}{\mathcal{N}}
\newcommand{\calE}{\mathcal{E}}
\newcommand{\calY}{\mathcal{Y}}
\newcommand{\leqs}{\leq_\mathrm{s}}
\newcommand{\vc}{\vec{c}}
\newcommand{\vp}{\vec{p}}
\newcommand{\vi}{\vec{i}}
\newcommand{\vx}{\vec{x}}
\newcommand{\vy}{\vec{y}}
\newcommand{\vw}{\vec{w}}
\newcommand{\vz}{\vec{z}}
\newcommand{\vu}{\vec{u}}
\newcommand{\vv}{\vec{v}}
\newcommand{\oX}{\overline{X}}
\renewcommand{\vb}{\vec{b}}
\newcommand{\bfr}{\mathbf{r}}
\newcommand{\bfp}{\mathbf{p}}

\def\longrightharpoonup{\relbar\joinrel\rightharpoonup}
\def\longleftharpoondown{\leftharpoondown\joinrel\relbar}
\def\longrightleftharpoons{\mathop{\vcenter{\hbox{\ooalign{\raise1pt\hbox{$\longrightharpoonup\joinrel$}\crcr\lower1pt\hbox{$\longleftharpoondown\joinrel$}}}}}}
\def\rxn{\mathop{\rightarrow}\limits}  
\def\lrxn{\mathop{\longrightarrow}\limits}
\def\revrxn{\mathop{\rightleftharpoons}\limits}
\def\lrevrxn{\mathop{\longrightleftharpoons}\limits}

\maketitle

\begin{abstract}

Chemical reaction networks (CRNs) formally model chemistry in a well-mixed solution.
CRNs are widely used to describe information processing occurring in natural cellular regulatory networks,
and with upcoming advances in synthetic biology, CRNs are a promising language for the design of artificial molecular control circuitry.
Nonetheless, despite the widespread use of CRNs in the natural sciences,
the range of computational behaviors exhibited by CRNs is not well understood.

CRNs have been shown to be efficiently Turing-universal when allowing for a small probability of error.
CRNs that are guaranteed to converge on a correct answer, on the other hand, have been shown to decide only the semilinear predicates.
We introduce the notion of function, rather than predicate, computation by representing the output of a function $f:\N^k\to\N^l$ by a count of some molecular species, i.e., if the CRN starts with $x_1,\ldots,x_k$ molecules of some ``input'' species $X_1,\ldots,X_k$, the CRN is guaranteed to converge to having $f(x_1,\ldots,x_k)$ molecules of the ``output'' species $Y_1,\ldots,Y_l$.
We show
that a function $f:\N^k \to \N^l$ is deterministically computed by a CRN if and only if its graph $\{ (\vx,\vy) \in \N^k \times \N^l\ |\ f(\vx) = \vy \}$ is a semilinear set.

Finally, we show that each semilinear function $f$ can be computed by a CRN on input $\vx$ in expected time $O(\mathrm{polylog}\ \|\vx\|_1)$.
\end{abstract}



\section{Introduction}
\label{sec-intro}

The engineering of complex artificial molecular systems will require a sophisticated understanding of how to \emph{program chemistry}.
A natural language for describing the interactions of molecular species in a well-mixed solution is that of (finite) chemical reaction networks (CRNs), i.e., finite sets of chemical reactions such as $A + B \to A + C$.
When the behavior of individual molecules is modeled, CRNs are assigned semantics through \emph{stochastic chemical kinetics}~\cite{Gillespie77}, in which reactions occur probabilistically with rate proportional to the product of the molecular count of their reactants and inversely proportional to the volume of the reaction vessel.

Traditionally CRNs have been used as a descriptive language to analyze naturally occurring chemical reactions
(as well as numerous other systems with a large number of interacting components such as
gene regulatory networks and animal populations).
However, recent investigations have viewed CRNs as a programming language for engineering artificial systems.
These works have shown CRNs to have eclectic computational abilities.
Researchers have investigated the power of CRNs to simulate Boolean circuits~\cite{magnasco1997chemical}, neural networks~\cite{hjelmfelt1991chemical}, and digital signal processing~\cite{riedelDSP2012}.
Other work has shown that bounded-space Turing machines can be simulated with an arbitrarily small, non-zero probability of error by a CRN with only a polynomial slowdown \cite{angluin2006fast}.\footnote{
This is surprising since finite CRNs necessarily must represent binary data strings in a unary encoding, since they lack positional information to tell the difference between two molecules of the same species.}
Even Turing universal computation is possible with an arbitrarily small, non-zero probability of error over all time~\cite{SolCooWinBru08}.
The computational power of CRNs also provides insight on why it can be computationally difficult to simulate them~\cite{soloveichik2009robust},
and why certain questions are frustatingly difficult to answer (e.g.\ undecidable)\cite{CooSolWinBru09,zavattaro2008termination}.
The programming approach to CRNs has also, in turn, resulted in novel insights regarding natural cellular regulatory networks~\cite{cardelli2012cell}.

Recent work proposes concrete chemical implementations of arbitrary CRNs, particularly using nucleic-acid strand-displacement cascades as the physical reaction primitive~\cite{SolSeeWin10,cardelli2011strand}.
Thus, since in principle any CRN can be built, hypothetical CRNs with interesting behaviors are becoming of more than theoretical interest.
One day artificial CRNs may underlie embedded controllers for biochemical, nanotechnological, or medical applications, where environments are inherently incompatible with traditional electronic controllers.

One of the best-characterized computational abilities of CRNs is the deterministic computation of predicates as investigated by Angluin, Aspnes and Eisenstat~\cite{AngluinAE06}.
(They considered an equivalent distributed computing model known as \emph{population protocols}.)
Some CRNs, when started in an initial configuration assigning nonnegative integer counts to each of $k$ different input species, are guaranteed to converge on a single ``yes'' or ``no'' answer, in the sense that there are two special ``voting'' species $L^1$ and $L^0$ so that eventually either $L^1$ is present and $L^0$ absent to indicate ``yes'', or vice versa to indicate ``no.''
The set of inputs $S \subseteq \N^k$ that cause the system to answer ``yes'' is then a representation of the decision problem solved by the CRN.
Angluin, Aspnes and Eisenstat showed that the input sets $S$ decidable by some CRN are precisely the \emph{semilinear} subsets of $\N^k$ (see below).

We extend these prior investigations of decision problems or predicate computation to study deterministic \emph{function} computation.
Consider the three examples in Fig.~\ref{fig:examples}(top).
These CRNs have the property that they converge to the right answer no matter the order in which the reactions happen to occur and are thus insensitive to stochastic effects as well as reaction rate constants.
Formally, we say a function $f:\N^k\to\N^l$ is computed by a CRN $\calC$ if the following is true.
There are ``input'' species $X_1,\ldots,X_k$ and ``output'' species $Y_1,\ldots,Y_l$ such that, if $\calC$ is initialized with $x_1,\ldots,x_k$ copies of $X_1,\ldots,X_k$, then it is guaranteed to reach a configuration in which the counts of $Y_1,\ldots,Y_l$ are described by the vector $f(x_1,\ldots,x_k)$, and these counts never again change.
For example, the CRN $\calC$ with the single reaction $X \to 2Y$ computes the function $f(x)=2x$ in the sense that, if $\calC$ starts in an initial configuration with $x$ copies of $X$ and 0 copies of $Y$, then $\calC$ is guaranteed to stabilize to a configuration with $2x$ copies of $Y$ (and no copies of $X$).
Similarly, the function $f(x) = \floor{x/2}$ is computed by the single reaction $2X \to Y$ (Fig.~\ref{fig:examples}(a)), in that the final configuration is guaranteed to have exactly $\floor{x/2}$ copies of $Y$ (and 0 or 1 copies of $X$, depending on whether $x$ is even or odd).

\begin{figure}[t]
\centering
\scalebox{0.55}{\includegraphics{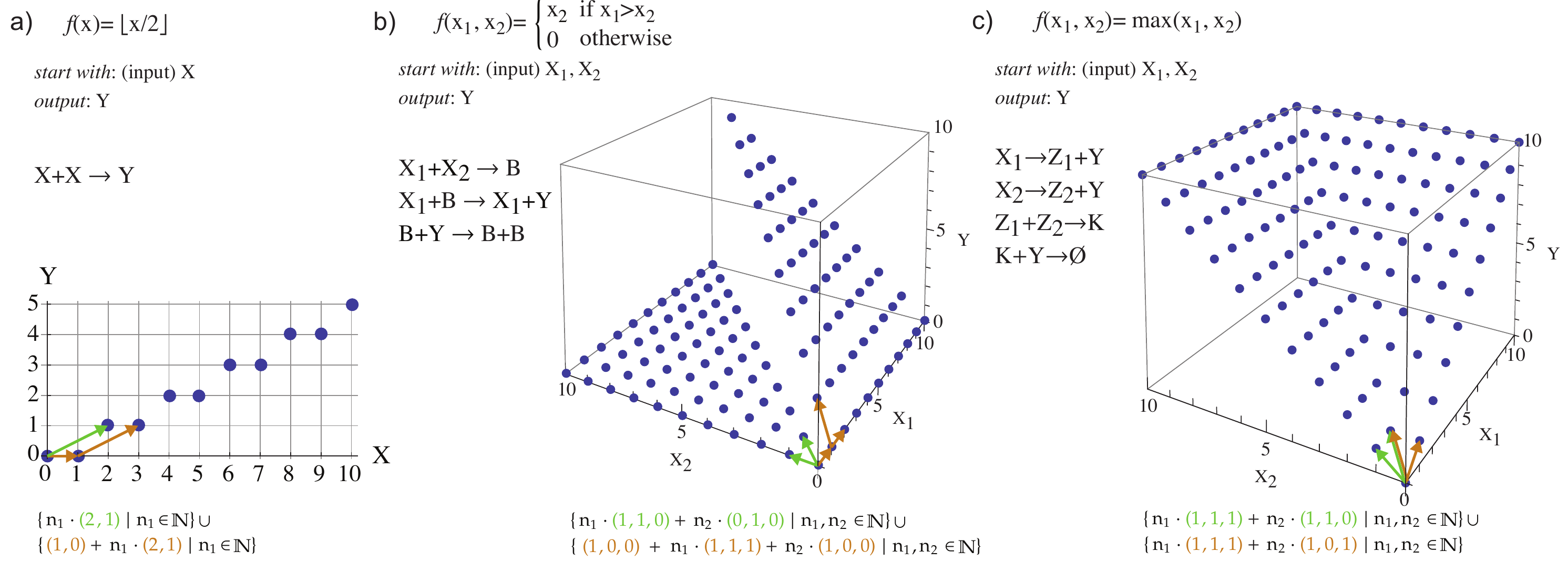}}
\caption{\footnotesize
Examples of deterministically computable functions.
(Top) Three functions and examples of CRNs deterministically computing them.
The input is represented in the molecular count of $X$ (for (a)), and moleculer counts of $X_1$, $X_2$ (for (b) and (c)).
The output is represented by the molecular count of $Y$.
Example (a) computes via the relative stoichiometry of reactants and products of a single reaction.
In example (b),
the second and third reactions convert $B$ to $Y$ and vice versa, catalyzed by $X_1$ and $B$, respectively.
Thus, if there are any $X_1$ remaining after the first reaction finishes (and thus $x_1 > x_2$),
all of $B$ can get converted to $Y$ permanently (since some $B$ is required to convert $Y$ back to $B$).
Since in this case the first reaction produces $x_2$ molecules of $B$,
$x_2$ molecules of the output $Y$ are eventually produced.
If the first reaction consumes all of $X_1$ (and thus $x_1 \leq x_2$),
eventually any $Y$ that was produced in the second reaction gets converted to $B$ by the third reaction.
To see that the CRN in (c) correctly computes the maximum, note that the first two reactions eventually produce $x_1+x_2$ molecules of $Y$, while the third reaction eventually produces $\min(x_1,x_2)$ molecules of $K$.
Thus the last reaction eventually consumes $\min(x_1,x_2)$ molecules of $Y$ leaving $x_1 + x_2 - \min(x_1,x_2) = \max(x_1,x_2)$ $Y$'s.
(Bottom) Graphs of the three functions.
The set of points belonging to the graph of each of these functions is a semilinear set.
Under each plot this semilinear set is written in the form of a union of linear sets corresponding to Equation~\ref{eq:linear}.
The defining vectors are shown as colored arrows in the graph.
\label{fig:examples}
}
\end{figure}


It is illuminating to compare the computation of division by $2$ shown in Fig.~\ref{fig:examples}(a) with another reasonable alternative:
reactions $X \to Y$ and $Y \to X$ (i.e.\ the reversible reaction $X \revrxn Y$).
If the rate constants of the two reactions are equal, the system equilibrium is at half of the initial amount of $X$ transformed to $Y$.
There are two stark differences between this implementation and that of Fig.~\ref{fig:examples}(a).
First, this CRN would not have an exact output count of $Y$, but rather a distribution around the equilibrium. (However, in the limit of large numbers, the error as a fraction of the total would converge to zero.)
The second difference is that the equilibrium amount of $Y$ for any initial amount of $X$ depends on the relative rate constants of the two reactions.
In contrast, the deterministic computation discussed in this paper relies on the identity and stoichiometry of the reactants and products rather than the rate constants.
While the rates of reactions are analog quantities, the identity and stoichiometry of the reactants and products are naturally digital.
Methods for physically implementing CRNs naturally yield systems with digital stoichiometry that can be set exactly~\cite{SolSeeWin10,cardelli2011strand}.
While rate constants can be tuned, being analog quantities, it cannot be expected that they can be controlled precisely.

A few general properties of this type of deterministic computation can be inferred.
The first property is that a deterministic CRN is able to handle input molecules added at any time, and not just initially.
Otherwise, if the CRN could reach a state after which it no longer ``accepts input'',
then there would be a sequence of reactions that would lead to an incorrect output even if all input is present initially.
This reaction sequence is one in which some input molecules remain unreacted while the CRNs goes to a state in which input is no longer accepted -- which is always possible.

The second general property of deterministic computation relates to composition.
As any bona fide computation must be composable, it is important to ask: can the output of one deterministic CRN be the input to another?
This is more difficult than in standard computing since there is in general no way of knowing when a CRN is done computing, or whether it will change its answer in the future.
This is essentially because a CRN cannot deterministically detect the absence of a species, and thus, for example, cannot discern when all input has been read.
Moreover, simply concatenating  two deterministic CRNs (renaming species to avoid conflict) does not always yield a deterministic CRN.
For example, consider computing the function $f(x_1, x_2) = \lfloor \max(x_1, x_2) / 2 \rfloor$ by composing the CRNs in Fig.~\ref{fig:examples}(c) and (a).
The new CRN is:
\begin{eqnarray*}
X_1 & \rxn & Z_1 + W \\
X_2 & \rxn & Z_2 + W \\
Z_1 + Z_2 & \rxn & K \\
K + W & \rxn & \emptyset \\
W + W & \rxn & Y
\end{eqnarray*}
where $W$ is the output species of the max computation, that acts as the input to the division by 2 computation.
Note that if $W$ happens to be converted to $Y$ by the last reaction before it reacts with $K$,
then the system can converge to a final output value of $Y$ that is larger than expected.
In other words, because the first CRN needs to consume its output $W$,
the second CRN can interfere by consuming $W$ itself (in the process of reading it out).

Unlike in the above example, two deterministic CRNs \emph{can} be simply concatenated to make a new deterministic CRN if the first CRN never consumes its output species (i.e.\ it produces its output ``monotonically'').
Since it doesn't matter when the input to the second CRN is produced (the first property, above), the overall computation will be correct.
Yet deterministically computing a non-monotonic function without consuming output species is impossible because the CRN must be able to handle some input molecules reacting only after the output has already been produced (i.e.\ the first property, above).
In a couple of places in this paper, we convert a non-monotonic function into a monotonic one over more outputs, to allow the result to be used by a downstream CRN (see below).

What do the functions in Fig.~\ref{fig:examples}(top) have in common such that the CRNs computing them can inevitably progress to the right answer no matter what order the reactions occur in?
What other functions can be computed similarly?
Answering these questions may seem difficult because it appears like the three examples, although all deterministic,
operate on different principles and seem to use different ideas.

We show that the functions deterministically computable by CRNs are precisely the semilinear functions,
where we define a function to be \emph{semilinear} if its \emph{graph} $\{ (\vx,\vy) \in \N^k \times \N^l\ |\ f(\vx) = \vy \}$ is a semilinear subset of $\N^k \times \N^l$.
(See Fig.~\ref{fig:examples}(bottom) for the graphs of the three example functions.)
This means that the graph of the function is a union of a finite number of \emph{linear} sets -- i.e.\ sets that can be written in the form
\begin{equation}  \label{eq:linear}
\setl{\vec{b} + n_1 \vec{u}_1 + \ldots + n_p \vec{u}_p}{n_1,\ldots,n_p\in\N}
\end{equation}
for some  fixed vectors
$\vec{b},\vec{u}_1,\ldots,\vec{u}_p \in\N^{k+l}$.
Fig.~\ref{fig:examples}(bottom) shows the graphs of the three example functions expressed as a union of sets of this form.
Informally, semilinear functions can be thought of as  ``piecewise linear functions'' with a finite number of pieces, and linear domains of each piece.\footnote{Semilinear sets have a number of characterizations.
They are often thought of as generalizations of arithmetic progressions.
They are also exactly the sets that are definable in Presburger arithmetic~\cite{Presburger30}: the first-order theory of the natural numbers with addition.
Equivalently, they are  the sets accepted by boolean combinations of ``modulo" and ``threshold" predicates~\cite{AngluinAE06}.
Semilinear functions are less well-studied.
The ``piecewise linear'' intuitive characterization  is  formalized in Lemma~\ref{lem-semilinear-function-finite-union-linear}.
}

This characterization implies, for example, that such functions as $f(x_1,x_2) = x_1 x_2$,  $f(x) = x^2$, or $f(x) = 2^x$ are not deterministically computable.
For instance, the graph of the function $f(x_1,x_2) = x_1 x_2$ consists of infinitely many lines of different slopes, and thus, while each line is a linear set, the graph is not a finite union of linear sets.
Our result employs the predicate computation characterization of Angluin, Aspnes and Eisenstat~\cite{AngluinAE06}, together with some nontrivial additional technical machinery.



\begin{figure}[t]
\centering
\scalebox{0.45}{\includegraphics{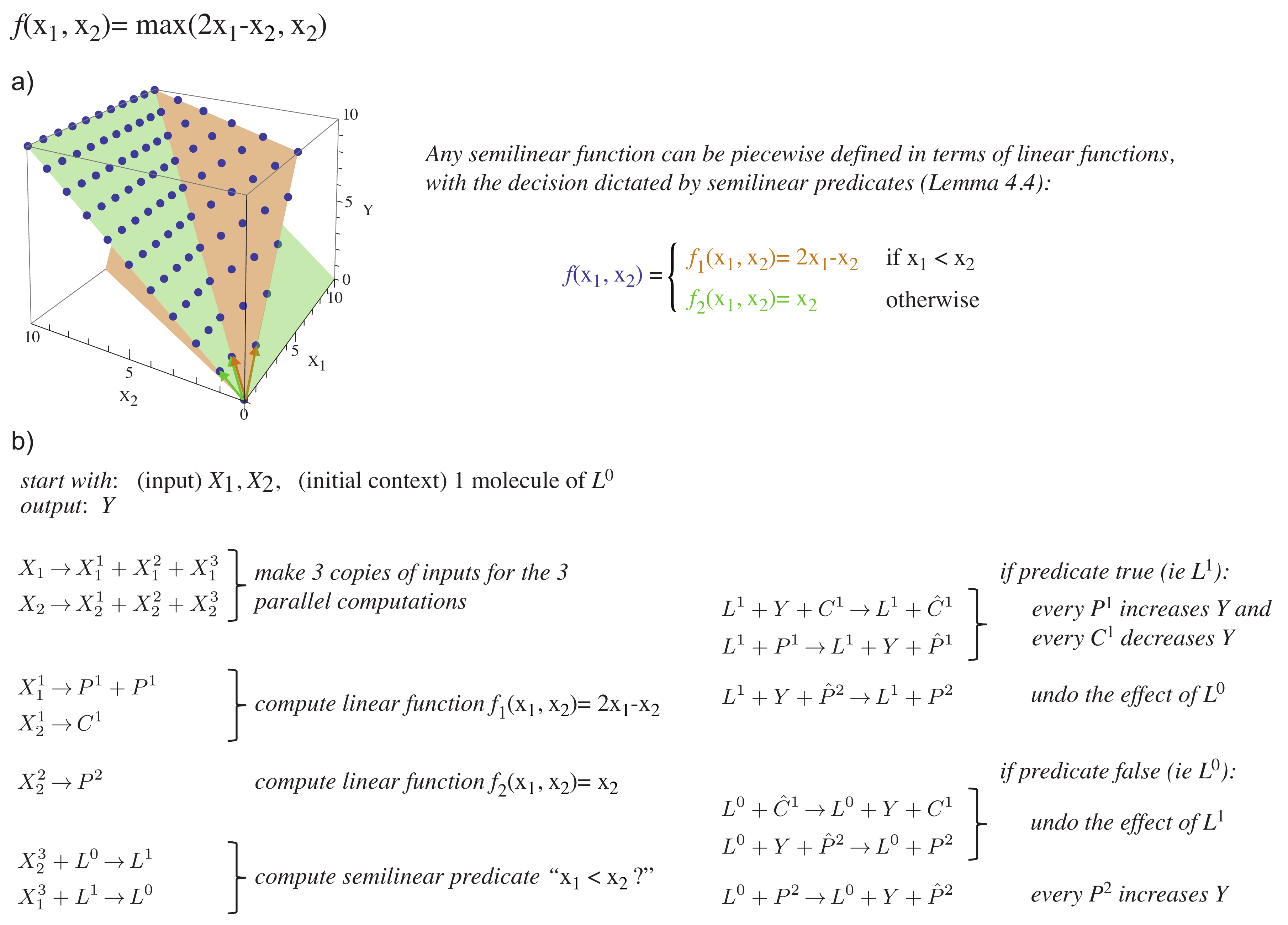}}
\caption{\footnotesize
An example capturing the essential elements of our systematic construction for computing semilinear functions (Lemma \ref{lem-compute-semilinear-n-log-n}).
To compute the target semilinear function,
we recast it as a piecewise function defined in terms of linear functions,
such that semilinear predicates can decide which of the linear functions is applicable for a given input (this recasting is possible by Lemma~\ref{lem-semilinear-function-finite-union-linear}).
(a) The graph of the target function visualizing the decomposition into linear functions.
(b) A CRN deterministically computing the target function with intuitive explanations of the reactions.
We use tri-molecular reactions for simplicity of exposition;
however, these can be converted into a sequence of bimolecular reactions.
Note that we allow an ``initial context": a fixed set of molecules that are always present in the initial state in addition to the input.
The linear functions $f_1$ and $f_2$ are computed monotonically by representing the output as the difference of $P$ (``produce'') minus $C$ (``consume'') species.
Thus although $P^1 - C^1$ could be changing non-monotonically, $P^1$ and $C^1$ do not decrease over time, allowing them to be used as inputs for downstream computation.
To compute the semilinear predicate ``$x_1 < x_2$?'', a single molecule, converted between $L^0$ and $L^1$ forms, goes back and forth consuming $X_1^3$ and $X_2^3$.
Whether it gets stuck in the $L^0$ or $L^1$ forms indicates the excess of $X_1^3$ or $X_2^3$.
The reactions in the right column use the output of this predicate computation to set the count of $Y$ (the global output) to either the value computed by $f_1$ or $f_2$.
Note that the CRN cannot ``know'' when the predicate computation has finished since the absence $X_1^3$ or $X_2^3$ cannot be detected.
Thus the reactions in the right column must be capable of responding to a change in $L^0/L^1$.
Species $\hat{P}^1$, $\hat{P}^2$, and  $\hat{C}^1$ are used to backup the values of $P^1$, $P^2$, and  $C^1$, enabling the switch in output.
\label{fig:systematicexample}
}
\end{figure}

While the example CRNs in Fig.~\ref{fig:examples} all seem to use different ``tricks'',
in Section~\ref{sec:detfuncspeed} we develop a systematic construction for any semilinear function.
To get the gist of this construction see the example in Fig.~\ref{fig:systematicexample}.
To obtain a CRN computing the example semilinear function $f(x_1, x_2) = \max(2 x_1 - x_2, x_2)$, we decompose the function into ``linear'' pieces: $f_1(x_1, x_2) = 2 x_1 - x_2$ and $f_2(x_1, x_2) = x_2$ (formally partial affine functions, see Section~\ref{sec-prelim}).
Then semilinear predicate computation (per Angluin, Aspnes and Eisenstat) is used to decide which linear function should be applied to a given input.
A decomposition compatible with this approach is always possible by Lemma~\ref{lem-semilinear-function-finite-union-linear}.
Linear functions such as $f_1$ and $f_2$ are easy for CRNs to deterministically compute by the relative stoichiometry of the reactants and products (analogously to the example in Fig.~\ref{fig:examples}(a)).
However, note that to correctly compose the computation of $f_1$ with the downstream computation (Fig.~\ref{fig:examples}(b), right column) we convert $f_1$ from a non-monotonic function with one output, to a monotonic function with two outputs such that the original output is encoded by their difference.

In the last part of this paper, we turn our attention to the time required for CRNs to converge to the answer. 
We show that every semilinear function can be deterministically computed on input $\vx$ in expected time $\mathrm{polylog}(\|\vx\|)$.
This is done by a similar technique used by Angluin, Aspnes, and Eisenstat~\cite{AngluinAE06} to show the equivalent result for predicate computation.
They run a slow deterministic computation in parallel with a fast randomized computation, allowing the deterministic computation to compare the two answers and update the randomized answer only if it is incorrect, which happens with low probability.
However, novel techniques are required since it is not as simple to ``nondestructively compare'' two integers (so that the counts are only changed if they are unequal) as to compare two Boolean values.



\section{Preliminaries}
\label{sec-prelim}


Given a vector $\vx\in\N^k$, let $\|\vx\| = \|\vx\|_1 = \sum_{i=1}^k |\vx(i)|$, where $\vx(i)$ denotes the $i$th coordinate of $\vx$.
A set $A \subseteq \N^k$ is \emph{linear} if there exist vectors $\vec{b},\vec{u}_1,\ldots,\vec{u}_p \in\N^k$ such that
$$A=\setl{\vec{b} + n_1 \vec{u}_1 + \ldots + n_p \vec{u}_p}{n_1,\ldots,n_p\in\N}.$$
$A$ is \emph{semilinear} if it is a finite union of linear sets.
If $f:\N^k\to\N^l$ is a function, define the \emph{graph} of $f$ to be the set
$\setl{ (\vx,\vy) \in \N^k \times \N^l }{ f(\vx) = \vy }.$
A function is \emph{semilinear} if its graph is a semilinear set.

We say a partial function $f:\N^k \dashrightarrow \N^l$ is \emph{affine} if there exist $kl$ rational numbers $a_{1,1},\ldots,a_{k,l}\in\Q$ and $l+k$ nonnegative integers $b_1,\ldots,b_l,c_1,\ldots,c_k\in\N$ such that, if $\vy = f(\vx)$, then for each $j \in \{1,\ldots,l\}$, $\vy(j) = b_j + \sum_{i=1}^k a_{i,j} (\vx(i) - c_i)$, and for each $i \in \{1,\ldots,k\}$, $\vx(i) - c_i \geq 0$.
(In matrix notation, there exist a $k \times l$ rational matrix $\vec{A}$ and vectors $\vb \in \N^l$ and $\vec{c} \in \N^k$ such that $f(\vx) = \vec{A} (\vx-\vec{c}) + \vb$.)
In other words, the graph of $f$, when projected onto the $(k+1)$-dimensional space defined by the $k$ coordinates corresponding to $\vx$ and the single coordinate corresponding to $\vy(j)$, is a subset of a $k$-dimensional hyperplane.

Four aspects of the definition of affine function invite explanation.

First, we allow partial functions because Lemma~\ref{lem-semilinear-function-finite-union-linear} characterizes the semilinear functions as finite combinations of affine functions, where the union of the domains of the functions is the entire input space $\N^k$.
The value of an affine function on an input outside of its domain is irrelevant (and in fact may be non-integer).

Second, we have two separate ``constant offsets'' $b_j$ and $c_i$.
Affine functions over the reals are typically defined with only one of these, $b_j$, where a function $f:\R^k\to\R^l$ is affine if there is a $k \times l$ real matrix $\vec{A}$ and vector $\vb \in \R^l$ such that $f(\vx) = \vec{A} \vx + \vb$.
If instead real affine functions were defined as $f(\vx) = \vec{A} (\vx - \vec{c}) + \vb$, one could re-write this as $f(\vx) = \vec{A} \vx - \vec{A} \vec{c} + \vb$ and, letting $\vb' = - \vec{A} \vec{c} + \vb$, obtain an affine function according to the former definition.
However, if we take this approach in dealing with integers, it may be that while $\vec{A} \vx - \vec{A} \vec{c} + \vb$ is integer-valued, the terms $\vec{A} \vx$ and $- \vec{A} \vec{c} + \vb$ are non-integer vectors, and when we compute affine functions with chemical reaction networks, these terms are handled separately by integer-valued counts of molecules.

Third, it may seem overly restrictive to require $b_j$ and $c_i$ to be nonnegative.
In fact, our proof of Lemma~\ref{lem-compute-linear-fast} is easily modified to show how to construct a CRN to compute an affine function that allows negative values for $b_j$ and $c_i$.
However, Lemma~\ref{lem-semilinear-function-finite-union-linear} shows that, when the function is such that its graph is a nonnegative linear set, then we may freely assume that $b_j$ and $c_i$ to be nonnegative.
Since this simplifies some of our definitions, we use this convention, although it is not a strictly necessary assumption to prove computability of affine functions by chemical reaction networks.

Fourth, the requirement that $\vx(i) - c_i \geq 0$ seems artificial.
When we prove that every semilinear function can be written as a finite union of partial affine functions with linear graphs (Lemma~\ref{lem-semilinear-function-finite-union-linear}), however, this will follow from the fact that the ``offset vector'' in the definition of a linear set is required to be nonnegative.

Note that by appropriate integer arithmetic, a partial function $f:\N^k \dashrightarrow \N^l$ is affine if and only if there exist $kl$ integers $n_{1,1},\ldots,n_{k,l}\in\Z$ and $2l+k$ nonnegative integers $b_1,\ldots,b_l,c_1,\ldots,c_k,d_1,\ldots,d_l\in\N$ such that, if $\vy = f(\vx)$, then for each $j \in \{1,\ldots,l\}$,
$\vy(j) = b_j + \frac{1}{d_j} \sum_{i=1}^k n_{i,j} (\vx(i) - c_i)$, and for each $i \in \{1,\ldots,k\}$, $\vx(i) - c_i \geq 0$.
Each $d_j$ may be taken to be the least common multiple of the denominators of the rational coefficients in the original definition.
We employ this latter definition, since it is more convenient for working with integer-valued molecular counts.

\subsection{Chemical reaction networks}

If $\Lambda$ is a finite set (in this paper, of chemical species), we write $\N^\Lambda$ to denote the set of functions $f:\Lambda \to \N$.
Equivalently, we view an element $\vc\in\N^\Lambda$ as a vector of $|\Lambda|$ nonnegative integers, with each coordinate ``labeled'' by an element of $\Lambda$.
Given $X\in \Lambda$ and $\vc \in \N^\Lambda$, we refer to $\vc(X)$ as the \emph{count of $X$ in $\vc$}.
We write $\vc \leq \vc'$ to denote that $\vc(X) \leq \vc'(X)$ for all $X \in \Lambda$.
Given $\vc,\vc' \in \N^\Lambda$, we define the vector component-wise operations of addition $\vc+\vc'$, subtraction $\vc-\vc'$, and scalar multiplication $n \vc$ for $n \in \N$.
If $\Delta \subset \Lambda$, we view a vector $\vc \in \N^\Delta$ equivalently as a vector $\vc \in \N^\Lambda$ by assuming $\vc(X)=0$ for all $X \in \Lambda \setminus \Delta.$

Given a finite set of chemical species $\Lambda$, a \emph{reaction} over $\Lambda$ is a triple $\alpha = \langle \bfr,\bfp,k \rangle \in \N^\Lambda \times \N^\Lambda \times \R^+$, specifying the stoichiometry of the reactants and products, respectively, and the \emph{rate constant} $k$.
If not specified, assume that $k=1$ (this is the case for all reactions in this paper), so that the reaction $\alpha=\langle \bfr,\bfp,1 \rangle$ is also represented by the pair $\pair{\bfr}{\bfp}.$
For instance, given $\Lambda=\{A,B,C\}$, the reaction $A+2B \to A+3C$ is the pair $\pair{(1,2,0)}{(1,0,3)}.$
A \emph{(finite) chemical reaction network (CRN)} is a pair $N=(\Lambda,R)$, where $\Lambda$ is a finite set of chemical \emph{species},
and $R$ is a finite set of reactions over $\Lambda$.
A \emph{configuration} of a CRN $N=(\Lambda,R)$ is a vector $\vc \in \N^\Lambda$.
We also write $\#_\vc X$ to denote $\vc(X)$, the \emph{count} of species $X$ in configuration $\vc$, or simply $\# X$ when $\vc$ is clear from context.

Given a configuration $\vc$ and reaction $\alpha=\pair{\bfr}{\bfp}$, we say that $\alpha$ is \emph{applicable} to $\vc$ if $\bfr \leq \vc$ (i.e., $\vc$ contains enough of each of the reactants for the reaction to occur).
If $\alpha$ is applicable to $\vc$, then write $\alpha(\vc)$ to denote the configuration $\vc + \bfp - \bfr$ (i.e., the configuration that results from applying reaction $\alpha$ to $\vc$).
If $\vc'=\alpha(\vc)$ for some reaction $\alpha \in R$, we write $\vc \to_N \vc'$, or merely $\vc \to \vc'$ when $N$ is clear from context.
An \emph{execution} (a.k.a., \emph{execution sequence}) $\calE$ is a finite or infinite sequence of one or more configurations $\calE = (\vc_0, \vc_1, \vc_2, \ldots)$ such that, for all $i \in \{1,\ldots,|\calE|-1\}$, $\vc_{i-1} \to \vc_{i}$.
If a finite execution sequence starts with $\vc$ and ends with $\vc'$, we write $\vc \to_N^* \vc'$, or merely $\vc \to^* \vc'$ when the CRN $N$ is clear from context.
In this case, we say that $\vc'$ is \emph{reachable} from $\vc$.

Let $\Delta \subseteq \Lambda$.
We say that $\vp\in\N^{\Delta}$ is a \emph{partial configuration (with respect to $\Delta$)}.
We write $\vp = \vc \rest \Delta$ for any configuration $\vc$ such that $\vc(X)=\vp(X)$ for all $X \in \Delta$, and we say that $\vp$ is the \emph{restriction of $\vc$ to $\Delta$}.
Say that a partial configuration $\vp$ with respect to $\Delta$ is \emph{reachable} from configuration $\vc'$ if there is a configuration $\vc$ reachable from $\vc'$ and $\vp = \vc \rest \Delta$.
In this case, we write $\vc' \to^* \vp$.

Turing machines, for example, have different semantic interpretations depending on the computational task under study (deciding a language, computing a function, etc.).
Similarly, in this paper we use CRNs to decide subsets of $\N^k$ and to compute functions $f:\N^k\to\N^l$.
In the next two subsections we define two semantic interpretations of CRNs that correspond to these two tasks.

\subsection{Stable decidability of predicates} \label{subsec-prelim-pred}
We now review the definition of stable decidability of predicates introduced by Angluin, Aspnes, and Eisenstat~\cite{AngluinAE06}.%
\footnote{Those authors use the term ``stably \emph{compute}'', but we reserve the term ``compute'' to apply to the computation of non-Boolean functions.}
Intuitively, some species ``vote'' for a yes/no answer
and the system stabilizes to an output when a consensus is reached and it can no longer change its mind.
The determinism of the system is captured in that it is impossible to stabilize to an incorrect answer, and the correct stable output is always reachable.

A \emph{chemical reaction decider} (CRD) is a tuple $\calD=(\Lambda,R,\Sigma,\Upsilon,\phi,\sigma)$, where $(\Lambda,R)$ is a CRN, $\Sigma \subseteq \Lambda$ is the \emph{set of input species}, $\Upsilon \subseteq \Lambda$ is the set of \emph{voters}\footnote{The definitions of~\cite{AngluinAE06} assume that $\Upsilon = \Lambda$ (i.e., every species votes).
However, it is not hard to show that we may assume there are only two voting species, $L^0$ and $L^1$, so that $\# L^0 > 0$ and $\# L^1 = 0$ means that the CRD is answering ``no'', and $\# L^0 = 0$ and $\# L^1 > 0$ means that the CRD is answering ``yes.''
This convention will be more convenient in this paper.
}, $\phi:\Upsilon\to\{0,1\}$ is the \emph{(Boolean) output function}, and $\sigma \in \N^{\Lambda \setminus \Sigma}$ is the \emph{initial context}.
An input to $\calD$ will be a vector $\vi_0\in\N^\Sigma$ (equivalently, $\vi_0\in\N^k$ if we write $\Sigma = \{X_1,\ldots,X_k\}$ and assign $X_i$ to represent the $i$'th coordinate).
Thus a CRD together with an input vector defines an initial configuration $\vi$ defined by $\vi(X) = \vi_0(X)$ if $X \in \Sigma$, and $\vi(X) = \sigma(X)$ otherwise.
We say that such a configuration is a \emph{valid initial configuration}, i.e., $\vi \rest (\Lambda \setminus \Sigma) = \sigma$.
If we are discussing a CRN understood from context to have a certain initial configuration $\vi$, we write $\#_0 X$ to denote $\vi(X)$.

We extend $\phi$ to a partial function $\Phi:\N^\Lambda \dashrightarrow \{0,1\}$ as follows.
$\Phi(\vc)$ is undefined if either $\vc(X) = 0$ for all $X \in \Upsilon$, or if there exist $X_0,X_1 \in \Upsilon$ such that $\vc(X_0) > 0$, $\vc(X_1) > 0$, $\phi(X_0)=0$ and $\phi(X_1)=1$.
Otherwise, there exists $b\in\{0,1\}$ such that $(\forall X \in \Upsilon) (\vc(X)>0 \implies \phi(X)=b)$; in this case, the \emph{output} $\Phi(\vc)$ of configuration $\vc$ is $b$.

A configuration $\vc$ is \emph{output stable} if $\Phi(\vc)$ is defined and, for all $\vc'$ such that $\vc \to^* \vc'$, $\Phi(\vc') = \Phi(\vc)$.
We say a CRD $\calD$ \emph{stably decides} the predicate $\psi:\N^\Sigma \to \{0,1\}$ if,
for any valid initial configuration $\vi \in \N^\Lambda$ with $\vi \upharpoonright \Sigma = \vi_0$, for all configurations $\vc\in\N^\Lambda$,
$\vi \to^* \vc$ implies $\vc\to^*\vc'$ such that $\vc'$ is output stable and $\Phi(\vc') = \psi(\vi_0)$.
Note that this condition implies that no incorrect output stable configuration is reachable from $\vi$.
We say that $\calD$ \emph{stably decides} a set $A\in\N^k$ if it stably decides its indicator function.

The following theorem is due to Angluin, Aspenes, and Eisenstat~\cite{AngluinAE06}:

\begin{theorem}[\cite{AngluinAE06}] \label{thm-semilinear}
A set $A \subseteq \N^k$ is stably decidable by a CRD if and only if it is semilinear.
\end{theorem}

The model they use is defined in a slightly different way.
They study \emph{population protocols}, a distributed computing model in which a fixed-size set of agents, each having a state from a finite set, undergo successive pairwise interactions, the two agents updating their states upon interacting.
This is equivalent to chemical reaction networks in which all reactions have exactly two reactants and two products.


In fact, the forward direction of Theorem~\ref{thm-semilinear} (every stably decidable set is semilinear) holds even if stable computation is defined with respect to \emph{any} relation $\to^*$ on $\N^k$ that is reflexive, transitive, and ``respects addition'', i.e., [$(\forall \vc_1,\vc_2,\vx\in\N^k)\ (\vc_1 \to^* \vc_2) \implies (\vc_1 + \vx \to^* \vc_2 + \vx$)].
These properties can easily be shown to hold for the CRN reachability relation.
The third property, in particular, means that if some molecules $\vc_1$ can react to form molecules $\vc_2$, then it is possible for them to react in the presence of some extra molecules $\vx$, such that no molecules from $\vx$ react at all.


\subsection{Stable computation of functions}

We now define a notion of stable computation of \emph{functions} similar to those above for predicates.\footnote{The extension from Boolean predicates to functions described by Aspnes and Ruppert~\cite{aspnes2007introduction} applies only to finite-range functions, where one can choose $|\Lambda| \geq |Y|$ for output range $Y$.}
Intuitively, the inputs to the function are the initial counts of input species $X_1,\ldots,X_k$, and the outputs are the counts of ``output" species $Y_1,\ldots,Y_l$.
The system stabilizes to an output when the counts of the output species can no longer change.
Again determinism is captured in that it is impossible to stabilize to an incorrect answer and the correct stable output is always reachable.

Let $k,l\in\Z^+$.
A \emph{chemical reaction computer (CRC)} is a tuple $\calC=(\Lambda,R,\Sigma,\Gamma,\sigma)$, where $(\Lambda,R)$ is a CRN, $\Sigma \subset \Lambda$ is the \emph{set of input species}, $\Gamma \subset \Lambda$ is the \emph{set of output species}, such that $\Sigma \cap \Gamma = \emptyset$, $|\Sigma|=k$, $|\Gamma|=l$, and $\sigma \in \N^{\Lambda\setminus\Sigma}$ is the \emph{initial context}.
Write $\Sigma = \{X_1,X_2,\ldots,X_k\}$ and $\Gamma = \{Y_1,Y_2,\ldots,Y_l\}$.
We say that a configuration $\vc$ is \emph{output count stable} if, for every $\vc'$ such that $\vc \to^* \vc'$ and every $Y_i \in \Gamma$, $\vc(Y_i) = \vc'(Y_i)$ (i.e., the counts of species in $\Gamma$ will never change if $\vc$ is reached).
As with CRD's, we require initial configurations $\vi$ of $\calC$ with input $\vi_0\in\N^\Sigma$ to obey $\vi(X)=\vi_0(X)$ if $X\in\Sigma$ and $\vi(X)=\sigma(X)$ otherwise, calling them \emph{valid initial configurations}.
We say that $\calC$ \emph{stably computes} a function $f:\N^k\to\N^l$ if for any valid initial configuration $\vi \in \N^{\Lambda}$,
$\vi \to^* \vc$ implies $\vc\to^*\vc'$ such that $\vc'$ is an output count stable configuration with $f(\vi(X_1), \vi(X_2), \ldots, \vi(X_k)) = (\vc'(Y_1), \vc'(Y_2), \ldots, \vc'(Y_l))$.
Note that this condition implies that no incorrect output stable configuration is reachable from $\vi$.

As an example of a formally defined CRC consider the function $f(x) =   \floor{x/2}$ shown in Fig.~\ref{fig:examples}(a).
This function is stably computed by the CRC $(\Lambda,R,\Sigma,\Gamma,\sigma)$ where $(\Lambda,R)$ is the CRN consisting of a single reaction $2 X \to Y$, $\Sigma = \{X\}$ is the set of input species, $\Gamma = \{Y\}$ is the set of output species, and the initial context $\sigma$ is zero for all species in $\Lambda\setminus\Sigma$.
In Fig.~\ref{fig:examples}(b) the initial context $\sigma(N)=1$, and is zero for all other species in in $\Lambda\setminus\Sigma$.
\repds{}{In examples (a) and (b), there is at most one reaction that can happen in any reachable configuration.
In contrast, different reactions may occur next in (c).
However, from any reachable state, we can reach the output count stable configuration with the correct amount of $Y$, satisfying our definition of stable computation.}

In Sections~\ref{sec:detfunc} and \ref{sec:detfuncspeed} we will describe systematic (but much more complex) constructions for these and all functions with semilinear graphs.

\subsection{Fair execution sequences}

Note that by defining ``deterministic" computation in terms of certain states being reachable and others not, we cannot guarantee the system will get to the correct output \emph{for any possible execution sequence}.
For example suppose an adversary controls the execution sequence.
Then $\{X \to 2Y, A \to B, B \to A\}$ will not reach the intended output state $y = 2x$ if the adversary simply does not let the first reaction occur, always preferring the second or third.

Intuitively, in a real chemical mixture, the reactions are chosen randomly and not adversarially, and the CRN will get to the correct output.
In this section we follow Angluin, Aspnes, and Eisenstat~\cite{AngluinAE06} and define a combinatorial condition called fairness on execution sequences that captures what is minimally required of the execution sequence to be guaranteed that a stably deciding/computing CRD/CRC will reach the output stable state.
In the next section we consider the  kinetic model, which ascribes probabilities  to execution sequences.
The kinetic model also defines the time of reactions, allowing us to study the computational complexity of the CRN computation.
Note that in the kinetic model, if the reachable configuration space is bounded for any start configuration (i.e.\ if from any starting configuration there are finitely many configurations reachable) then any observed execution sequence will be fair with probability 1.
(This will be the case for our construction in Section~\ref{sec:detfuncspeed}.)

An infinite execution $\calE = (\vc_0,\vc_1,\vc_2,\ldots)$ is \emph{fair} if, for all partial configurations $\vp$, if $\vp$ is infinitely often reachable then it is infinitely often reached.\footnote{i.e.
$(\forall \Delta \subseteq \Lambda) (\forall \vp\in\N^\Delta) [((\exists^\infty i\in\N)\ \vc_i \to^* \vp) \implies ((\exists^\infty j\in\N)\ \vp = \vc_j \rest \Delta)].$}
In other words, no reachable partial configuration is ``starved''.%
\footnote{This definition of fairness is stricter than that used in~\cite{AngluinAE06},
which used only full configurations rather than partial configurations.
We choose this definition to prevent intuitively unfair executions from vacuously satisfying the definition of ``fair'' simply because of some species whose count is monotonically increasing with time (preventing any configuration from being infinitely often reachable).
Such a definition is unnecessary in~\cite{AngluinAE06} because population protocols by definition have a finite state space, since they enforce that every reaction has precisely two reactants and two products.}
This definition, applied to finite executions, deems all of them fair vacuously.
We wish to distinguish between finite executions that can be extended by applying another reaction and those that cannot.
Say that a configuration is \emph{terminal} if no reaction is applicable to it.
We say that a finite execution is \emph{fair} if and only if it 
ends in a terminal configuration.
For any species $A \in \Lambda$, we write $\#_\infty A$ to denote the eventual convergent count of $A$ if $\# A$ is guaranteed to stabilize on any fair execution sequence; otherwise, $\#_\infty A$ is undefined.


The next lemma characterizes stable computation of functions by CRCs in terms of fair execution sequences, showing that the counts of output species will converge to the correct output values on any fair execution sequence.
An analogous lemma holds for CRDs.

\begin{lem}
A CRC stably computes a function $f:\N^k\to\N^l$ if and only if for every valid initial configuration $\vi \in \N^{\Lambda}$, every fair execution $\calE = (\vi,\vc_1,\vc_2,\ldots)$ contains an output count stable configuration $\vc$ such that $f(\vi(X_1), \vi(X_2), \ldots, \vi(X_k)) = (\vc(Y_1), \vc(Y_2), \ldots, \vc(Y_l))$.
\end{lem}
\begin{proof}
The ``if'' direction follows because every finite execution sequence can be extended to be fair,
and thus an output count stable configuration with the correct output is always reachable.
The ``only if'' direction is shown as follows.
We know that from any reachable configuration $\vc$,
some correct output stable configuration $\vc'$ is reachable (but possibly different $\vc'$ for different $\vc$).
We'll argue that in any infinite fair execution sequence there is some partial configuration that is reachable infinitely often, and that any state with this partial configuration is the correct stable output state.
Consider an infinite fair execution sequence $\vc_1, \vc_2, \dots$, and the corresponding reachable correct output stable configurations $\vc'_1, \vc'_2, \dots$.
As in Lemma~11 of~\cite{AngluinAE06},
there is some integer $k \geq 1$ such that a configuration is output count stable
if and only if it is output count stable when each coordinate that is larger than $k$ is set to exactly $k$ ($k$-truncation).
The infinite sequence $\vc'_1, \vc'_2, \dots$ must have an infinite subsequence sharing the same k-truncation.
Let $\vp$ be the partial configuration consisting of the correct output and all the coordinates less than $k$ in the shared truncation.
This partial configuration is reachable infinitely often, and no matter what the counts of the other species outside of $\vp$ are, the resulting configuration is output count stable.
\qedl\end{proof}


\subsection{Kinetic model}

The following model of stochastic chemical kinetics is widely used in quantitative biology and other fields dealing with chemical reactions between species present in small counts~\cite{Gillespie77}.
It ascribes probabilities to execution sequences, and also defines the time of reactions, allowing us to study the computational complexity of the CRN computation in Section~\ref{sec:detfuncspeed}.

In this paper, the rate constants of all reactions are $1$, and we define the kinetic model with this assumption.
A reaction is \emph{unimolecular} if it has one reactant and \emph{bimolecular} if it has two reactants.
We use no higher-order reactions in this paper when using the kinetic model.

The kinetics of a CRN is described by a continuous-time Markov process as follows.
Given a fixed volume $v$ and current configuration $\vc$, the \emph{propensity} of a unimolecular reaction $\alpha : X \to \ldots$ in configuration $\vc$ is $\rho(\vc, \alpha) = \#_\vc X$.
The propensity of a bimolecular reaction $\alpha : X + Y \to \ldots$, where $X \neq Y$, is $\rho(\vc, \alpha) = \frac{\#_\vc X \#_\vc Y}{v}$.
The propensity of a bimolecular reaction $\alpha : X + X \to \ldots$ is $\rho(\vc, \alpha) = \frac{1}{2} \frac{\#_\vc X (\#_\vc X - 1)}{v}$.
The propensity function determines the evolution of the system as follows.
The time until the next reaction occurs is an exponential random variable with rate $\rho(\vc) = \sum_{\alpha \in R} \rho(\vc,\alpha)$ (note that $\rho(\vc)=0$ if no reactions are applicable to $\vc$).
The probability that next reaction will be a particular $\alpha_{\text{next}}$ is $\frac{\rho(\vc,\alpha_{\text{next}})}{\rho(\vc)}$.

The kinetic model is based on the physical assumption of well-mixedness valid in a dilute solution.
Thus, we assume the \emph{finite density constraint}, which stipulates that a volume required to execute a CRN must be proportional to the maximum molecular count obtained during execution~\cite{SolCooWinBru08}.
In other words, the total concentration (molecular count per volume) is bounded.
This realistically constrains the speed of the computation achievable by CRNs.
Note, however, that it is problematic to define the kinetic model for CRNs in which the reachable configuration space is unbounded for some start configurations, because this means that arbitrarily large molecular counts are reachable.\footnote{One possibility is to have a ``dynamically" growing volume as in~\cite{SolCooWinBru08}.}
We apply the kinetic model only to CRNs with  configuration spaces that are bounded for each start configuration.

\section{Exactly the semilinear functions can be deterministically computed}
\label{sec:detfunc}
In this section we use Theorem~\ref{thm-semilinear} to show that only ``simple'' functions can be stably computed by CRCs.
This is done by showing how to reduce the computation of a function by a CRC to the decidability of its graph by a CRD, and vice versa.
In this section we do not concern ourselves with kinetics.
Thus the volume is left unspecified,
and we consider the combinatorial-only condition of fairness on execution sequences\repds{.}{ for our positive result (Lemma~\ref{lem-pres-comp}) and direct reachability arguments for the negative result (Lemma~\ref{lem-comp-pres}).}


The next lemma shows that every function computable by a chemical reaction network is semilinear
\repds{}{by reducing stably deciding a set that is the graph of a function to stably computing that function.
It turns out that the reduction technique of introducing ``production" and ``consumption" indicator species will be a general technique, used repeatedly in this paper.
}

\begin{lem}\label{lem-comp-pres}
  Every function stably computable by a CRC is semilinear.
\end{lem}

\begin{proof}
\repds{}{
Suppose there is a CRC $\calC$ stably computing $f$. 
We will construct a CRD $\calD$ that stably decides the graph of $f$. 
By Theorem~\ref{thm-semilinear}, this implies that the graph of $f$ is semilinear.
Intuitively, the difficulty lies in checking whether the amount of the outputs $Y_i$ produced by $\calC$ matches the  value given to the decider $\calD$ as input.
What makes this non-trivial is that $\calD$ does not know whether $\calC$ has finished computing, and thus must compare $Y_i$ while $Y_i$ is potentially being changed by $\calC$.
In particular, $\calD$ cannot consume $Y_i$ or that could interfere with the operation of $\calC$.
}

  Let $\calC=(\Lambda,R,\Sigma,\Gamma,\sigma)$ be the CRC that stably computes $f:\N^k\to\N^l$, with input species $\Sigma = \{X_1,\ldots,X_k\}$ and output species $\Gamma = \{Y_1,\ldots,Y_l\}$.
  Modify $\calC$ to obtain the following CRD $\calD = (\Lambda',R',\Sigma',\Upsilon',\phi',\sigma')$.
\repds{
  Let $\calY^C = \{ Y_1^C, \ldots, Y_l^C \}$, where each $Y_i^C \not\in \Lambda$ are new species.
  Let $\calY^P = \{ Y_1^P, \ldots, Y_l^P \}$, where each $Y_i^P \not\in \Lambda$ are new species.
}{
  Let $\calY^C = \{ Y_1^C, \ldots, Y_l^C \}$ and $\calY^P = \{ Y_1^P, \ldots, Y_l^P \}$, where each $Y_i^C, Y_i^P \not\in \Lambda$ are new species.%
}
  Intuitively, $\# Y_i^P$ represents the number of $Y_i$'s produced by $\calC$ and $\# Y_i^C$ the number of $Y_i$'s consumed by $\calC$.
  The goal is for $\calD$ to stably decide the predicate $f(\#_0 X_1,\ldots,\#_0 X_k) = (\#_0 Y_1^C,\ldots,\#_0 Y_l^C)$. 
  In other words, the initial configuration of $\calD$ will be the same as that of $\calC$ except for some copies of $Y_i^C$, equal to the purported output of $f$ to be tested by $\calD$.
\repds{Since every predicate stably decidable by a CRD is semilinear (Theorem~\ref{thm-semilinear}), this will prove the lemma.}{}

  Let $\Lambda' = \Lambda \cup \calY^C \cup \calY^P \cup \{L^0,L^1\}$.
  Let $\Sigma' = \Sigma \cup \calY^C$.
  Let $\Upsilon' = \{L^0,L^1\}$, with $\phi(L^0) = 0$ and $\phi(L^1) = 1$.
  Let \repds{}{$\sigma'(L^1) = 1$ and} $\sigma'(S) = 0$ for all \repdd{other}{} $S \in \Lambda' \setminus $ \repdd{$\Sigma'$}{$(\Sigma' \cup \{L^1\})$}.
  Modify $R$ \repds{by adding reactions}{} to obtain $R'$ as follows.
  For each reaction $\alpha$ that consumes a net number $n$ of $Y_i$ molecules, append $n$ products $Y_i^C$ to $\alpha$.
  For each reaction $\alpha$ that produces a net number $n$ of $Y_i$ molecules, append $n$ products $Y_i^P$ to $\alpha$.
  For example, the reaction $A + 2B + Y_1 + 3Y_3 \to Z + 3Y_1 + 2 Y_3$ becomes $A + 2B + Y_1 + 3Y_3 \to Z + 3Y_1 + 2 Y_3 + 2Y_1^P + Y_3^C$.
\repds{Since $\calC$ is count-stable, eventually no reactions producing or consuming net copies of $Y_i$ are possible, whence $\calD$ as defined so far is count output stable with respect to $Y_i^P$ and $Y_i^C$ as well.}{}

  Then add the following additional reactions to $R'$, for each $i \in \{1,\ldots,l\}$,
  \begin{eqnarray}
    Y_i^P + Y_i^C &\to& L^1 \label{rxn1-eq}
    \\
    Y_i^P + L^1 &\to& Y_i^P + L^0 \label{rxn2-eq}
    \\
    Y_i^C + L^1 &\to& Y_i^C + L^0 \label{rxn3-eq}
    \\
    L^0 + L^1 &\to& L^1 \label{rxn4-eq}
  \end{eqnarray}

\repds{}{Observe that if $f(\#_0 X_1,\ldots,\#_0 X_k) = (\#_0 Y_1^C,\ldots,\#_0 Y_l^C)$,
then from any reachable configuration we can reach a configuration without any $Y_i^P$ or $Y_i^C$ for all $i$, and such that no more of either kind can be produced.
(The CRC stabilizes and all of $Y_i^P$ and $Y_i^C$ is consumed by reaction~\ref{rxn1-eq}.)
In this configuration we must have $\# L^1 > 0$ because the last instance of reaction \ref{rxn1-eq} produced it (or if no output was ever produced, $L^1$ comes from the initial context $\sigma'$),
and $L^1$ can no longer be consumed in reactions~\ref{rxn2-eq}--\ref{rxn3-eq}.
Thus, since all of $L^0$ can be consumed in reaction~\ref{rxn4-eq},
a configuration with $\# L^1 >0$ and $\# L^0 = 0$ is always reachable,
and this configuration is output stable.
}

\repds{}{Now suppose $f(\#_0 X_1,\ldots,\#_0 X_k) \neq (\#_0 Y_1^C,\ldots,\#_0 Y_l^C)$ for some \repdd{}{output} coordinate $i^*$ \repdd{}{$\in \{1,\ldots,l\}$}.
This means that from any reachable configuration we can reach a configuration with either $\# Y_{i^*}^P > 0$ or $\# Y_{i^*}^C > 0$ but not both,
and such that for all $i$, no more of $Y_i^P$ and $Y_i^C$ can be produced.
(\repdd{The}{This happens when the} CRC stabilizes and reaction~\ref{rxn1-eq} consumes the smaller of $Y_{i^*}^P$ or $Y_{i^*}^C$.)
From this configuration, we can reach a configuration with $\# L^0 > 0$ and $\# L^1 = 0$ through reactions~\ref{rxn2-eq}--\ref{rxn3-eq}.
This is an output stable configuration since reactions~\ref{rxn2-eq}--\ref{rxn4-eq} require $L^1$.
}
\qedl\end{proof}

\repds{  In the following, we use $\#_\infty^\uparrow Y_i^P$ to denote the total number of $Y_i^P$ ever produced and $\#_\infty^\uparrow Y_i^C$ to denote $\#_0 Y_i^C$ plus the total number of $Y_i^C$'s ever produced.
  Note that, if and only if $f(\#_0 X_1,\ldots,\#_0 X_k) = (\#_0 Y_1^C,\ldots,\#_0 Y_l^C)$, then eventually, for each $i$, $\# Y_i^P$ and $\# Y_i^C$ stabilize to equal values in the absence of reaction \eqref{rxn1-eq};
  in other words, if and only if $\#_\infty^\uparrow Y_i^P = \#_\infty^\uparrow Y_i^C$.
}{}

\repds{
  Since $Y_i^P$ and $Y_i^C$ are possibly produced but not consumed by reactions other than \eqref{rxn1-eq}, we may think of reaction~\eqref{rxn1-eq} as if it does not occur until $\# Y_i^P$ and $\# Y_i^C$ have stabilized, even though reaction~\eqref{rxn1-eq} may consume some copies of $Y_i^P$ and $Y_i^C$ before all eventual copies have been produced.
}{}

\repds{
  Reactions~\eqref{rxn1-eq}-\eqref{rxn4-eq} 
  ensure that if $\#_\infty^\uparrow Y_i^P = \#_\infty^\uparrow Y_i^C$ for all $i \in \{1,\ldots,l\}$, then $\#_\infty L^1 > 0$ and $\#_\infty L^0 = 0$, and if $\#_\infty^\uparrow Y_i^P \neq \#_\infty^\uparrow Y_i^C$ for some $i \in \{1,\ldots,l\}$, then $\#_\infty L^1 = 0$ and $\#_\infty L^0 > 0$.
  To show that this holds, we have two cases for each $i \in \{1,\ldots,l\}$.
  In the following, we write $f(\# X_1,\ldots,\# X_k)_i$ to denote the value $\# Y_i$ if $f(\# X_1,\ldots,\# X_k)=(\# Y_1,\ldots,\# Y_l)$.
}{}

The next lemma shows the converse of Lemma~\ref{lem-comp-pres}.
Intuitively, it uses a random search of the output space to look for the correct answer to the function and uses a predicate decider to check whether the correct solution has been found.

\begin{lem} \label{lem-pres-comp}
  Every semilinear function is stably computable by a CRC.
\end{lem}

\newcommand{\oF}{\widehat{F}}

\begin{proof}
  Let $f:\N^k \to \N^l$ be a semilinear function, and let
  $$
    F = \setl{ (\vx,\vy) \in \N^k \times \N^l }{ f(\vx) = \vy }
  $$
  denote the graph of $f$.
  We then consider the set
  $$
    \oF = \setl{ (\vx,\vy_P,\vy_C) \in \N^k \times \N^l \times \N^l }{ f(\vx) = \vy_P - \vy_C }.
  $$
  Intuitively, $\oF$ defines the same function as $F$, but with each output variable expressed as the difference between two other variables.
  Note that $\oF$ is not the graph of a function since for each $\vy \in \N^l$ there are an infinite number of pairs $(\vy_P,\vy_C)$ such that $\vy_P - \vy_C = \vy$.
  However, we only care that $\oF$ is a semilinear set so long as $F$ is a semilinear set, by Lemma~\ref{lem-semilinear-difference}, proven below.

  Then by Theorem~\ref{thm-semilinear}, $\oF$ is stably decidable by a CRD $\calD = (\Lambda,R,\Sigma,\Upsilon,\phi,\sigma)$, where $$\Sigma = \{X_1,\ldots,X_k,Y^P_1,\ldots,Y^P_l,Y^C_1,\ldots,Y^C_l\},$$ and we assume that $\Upsilon$ contains only species $L^1$ and $L^0$ such that for any output-stable configuration of $\calD$, exactly one of $\# L^1$ or $\# L^0$ is positive to indicate a yes or no answer, respectively.

  Define the CRC $\calC = (\Lambda',R',\Sigma',\Gamma',\sigma')$ as follows.
  Let $\Sigma' = \{X_1,\ldots,X_k\}$.
  Let $\Gamma' = \{Y_1,\ldots,Y_l\}$.
  Let $\Lambda' = \Lambda \cup \Gamma'$.
  Let $\sigma'(S) = \sigma(S)$ for all $S \in \Lambda \setminus \Sigma$, and let $\sigma'(S) = 0$ for all $S \in \Lambda' \setminus (\Lambda \setminus \Sigma)$.
  Intuitively, we will have $L^0$ change the value of $\vy$ (by producing either $Y^P_j$ or $Y^C_j$ molecules), since $L^0$'s presence indicates that $\calD$ has not yet decided that the predicate is satisfied.
  It essentially searches for new values of $\vy$ that do satisfy the predicate.
  This indirect way of representing the value $\vy$ is useful because $\vy_P$ and $\vy_C$ can both be increased monotonically to change $\vy$ in either direction.
  If $\calD$ had $Y_j$ as a species directly, and if we wanted to test a lower value of $\vy_j$, then this would require consuming a copy of $Y_j$, but this may not be possible if $\calD$ has already consumed all of them.

  Let $R'$ be $R$ plus the following reactions for each $j \in \{1, \ldots,l\}$:
  \begin{eqnarray}
    L^0 &\to& L^0 + Y_j^P + Y_j \label{rxn5-eq}
    \\
    L^0 + Y_j &\to& L^0 + Y_j^C \label{rxn6-eq}
  \end{eqnarray}

  It is clear that reactions \eqref{rxn5-eq} and \eqref{rxn6-eq} enforce that at any time, $\# Y_j$ is equal to the total number of $Y^P_j$'s produced by reaction \eqref{rxn5-eq} minus the total number of $Y^C_j$'s produced by reaction \eqref{rxn6-eq} (although some of each of $Y^P_j$ or $Y^C_j$ may have been produced or consumed by other reactions in $R$).

  Suppose that $f(\vx) \neq (\# Y_1,\ldots,\# Y_l)$.
  Then if there are no $L^0$ molecules present, the counts of $Y^P_j$ and $Y^C_j$ are not changed by reactions \eqref{rxn5-eq} and \eqref{rxn6-eq}.
  Therefore only reactions in $R$ proceed, and by the correctness of $\calD$, eventually an $L^0$ molecule is produced (since eventually $\calD$ must reach an output-stable configuration answering ``no'', although $L^0$ may appear before $\calD$ reaches an output-stable configuration, if some $L^1$ are still present).
  Once $L^0$ is present, by the fairness condition (choosing $\Delta = \{Y_1,\ldots.Y_l\}$), eventually the value of $(\# Y_1,\ldots,\# Y_l)$ will change by reaction \eqref{rxn5-eq} or \eqref{rxn6-eq}.
  In fact, \emph{every} value of $(\# Y_1,\ldots,\# Y_l)$ is possible to explore by the fairness condition.

  Suppose then that $f(\vx) = (\# Y_1,\ldots,\# Y_l)$.
  Perhaps $L^0$ is present because the reactions in $R$ have not yet reached an output-stable ``yes'' configuration.
  Then perhaps the value of $(\# Y_1,\ldots,\# Y_l)$ will change so that $f(\vx) \neq (\# Y_1,\ldots,\# Y_l)$.
  But by the fairness condition, a correct value of $(\# Y_1,\ldots,\# Y_l)$ must be present infinitely many times, so again by the fairness condition, since from such a configuration it is possible to eliminate all $L^0$ molecules before producing $Y^P_j$ or $Y^C_j$ molecules, this must eventually happen.
  When all $L^0$ molecules are gone while $f(\vx) = (\# Y_1,\ldots,\# Y_l)$ and $\calD$ is in an output-stable configuration (thus no $L^0$ can ever again be produced), then it is no longer possible to change the value of $(\# Y_1,\ldots,\# Y_l)$, whence $\calC$ has reached a count-stable configuration with the correct answer.
  Therefore $\calC$ stably computes $f$.
\qedl\end{proof}

\begin{lem} \label{lem-semilinear-difference}
  Let $k,l\in\Z^+$, and suppose $F \subseteq \N^k \times \N^l$ is semilinear.
  Define
  $$
    \oF = \setl{ (\vx,\vy_P,\vy_C) \in \N^k \times \N^l \times \N^l }{ (\vx,\vy_P - \vy_C) \in F }.
  $$
  Then $\oF$ is semilinear.
\end{lem}

\begin{proof}
  Let $F_1,\ldots,F_t$ be linear sets such that $F = \bigcup_{i=1}^t F_i$.
  For each $i\in\{1,\ldots,t\}$, define
  $$
    \oF_i = \setl{ (\vx,\vy_P,\vy_C) \in \N^k \times \N^l \times \N^l }{ (\vx,\vy_P - \vy_C) \in F_i }.
  $$
  It suffices to show that each $\oF_i$ is linear since $\oF = \bigcup_{i=1}^t \oF_i$.
  Let $i\in\{1,\ldots,t\}$ and let $\vb,\vu_1,\ldots,\vu_r\in \N^k \times \N^l$ be such that
  $$
    F_i = \setl{\vb + \sum_{j=1}^r n_j \vu_j}{ n_j \in \N}.
  $$
  Define the vectors
  $\vv_1,\ldots,\vv_r\in \N^k \times \N^l \times \N^l$
  as
  $\vv_j = (\vu_j,0^l)$.
  Here, $0^l$ denotes the vector in $\N^l$ consisting of all zeros.
  In other words, let $\vv_j$ be $\vu_j$ on its first $k+l$ coordinates and 0 on its last $l$ coordinates.
  Similarly define $\vb' = (\vb,0^l)$.

  Also, for each $j \in \{1,\ldots,l\}$ define
  $\vv_{r+j} = (0^{k},0^{j-1} 1 0^{l-j}, 0^{j-1} 1 0^{l-j}).$ (i.e., a single 1 in the position corresponding to the $j$th output coordinate, one for $\vy_P$ and one for $\vy_C$).
  Without the vectors $\vv_{r+j}$, the set of points defined by $\vb',\vv_1,\ldots,\vv_r$ would be simply $F_i$ with $l$ 0's appended to the end of each vector.
  By adding the vectors $\vv_{r+j}$, for each $(\vx,\vy)\in F_i$ and each $\vy_P,\vy_C \in \N^l$ such that $\vy = \vy_P - \vy_C$, we have that $(\vx,\vy_P,\vy_C) = \vb' + \sum_{j=1}^{r+1} n_j \vv_j$ for some $n_1,\ldots,n_{r+l} \in \N$; in particular, for $n_1,\ldots,n_r$ chosen such that $(\vx,\vy) = \vb + \sum_{j=1}^{r} n_j \vu_j$ and $n_{r+j} = \vy_C(j)$ for each $j\in\{1,\ldots,l\}$.

  Thus $\oF_i = \setl{\vb' + \sum_{j=1}^{r+l} n_j \vv_j}{ n_j \in \N},$ whence $\oF_i$ is linear.
\qedl\end{proof}

Lemmas~\ref{lem-comp-pres} and~\ref{lem-pres-comp} immediately imply the following theorem.

\begin{theorem} \label{thm-func}
  A function $f:\N^k\to\N^l$ is stably computable by a CRC if and only if it is semilinear.
\end{theorem}

One unsatisfactory aspect of Lemma~\ref{lem-pres-comp} is that we do not reduce the computation of $f$ directly to a CRD deciding the graph $F$ of $f$, but rather to $\calD$ deciding a related set $\oF$.
It is not clear how to directly reduce to a CRD deciding $F$ since it is not obvious how to modify such a CRD to monotonically produce extra species that could be processed by the CRC computing $f$.
Lemma~\ref{lem-comp-pres}, on the other hand, directly uses $\calC$ as a black-box.
Although we know that $\calC$, being a chemical reaction computer, is only capable of computing semilinear functions, if we imagine that some external powerful ``oracle'' controlled the reactions of $\calC$ to allow it to stably compute a non-semilinear function, then $\calD$ would decide that function's graph.
Thus Lemma~\ref{lem-comp-pres} is more like the black-box oracle Turing machine reductions employed in computability and complexity theory, which work no matter what mythical device is hypothesized to be responsible for answering the oracle queries.

\section{Semilinear functions can be quickly computed}
\label{sec:detfuncspeed}

\newcommand{\hvy}{\hat{\vy}}
\newcommand{\hY}{\widehat{Y}}
\newcommand{\oY}{\overline{Y}}

\newcommand{\Fix}{P_\textrm{fix}}
\newcommand{\Chck}{P_\textrm{check}}
\newcommand{\ERR}{\mathrm{ERR}}
\newcommand{\Exp}{\mathrm{E}}

\newcommand{\FW}{\mathrm{FW}}   
\newcommand{\TF}{\mathrm{TF}}   
\newcommand{\SD}{\mathrm{SD}} 

Lemma~\ref{lem-pres-comp} describes how a CRC can deterministically compute any semilinear function.
However, there are problems with this construction if we attempt to use it to evaluate the speed of semilinear function computation in the kinetic model.
First, the configuration space is unbounded for any input since the construction searches over outputs without setting bounds.
Thus, more care must be taken to ensure that any infinite execution sequence will be fair with probability 1 in the kinetic model.
What is more, since the maximum molecular count is unbounded, it is not clear how to set the volume for the time analysis.
Even if we attempt to properly define kinetics,  it seems like any reasonable time analysis of the random search process will result in expected time at least exponential in the size of the output.\footnote{The random walk is biased downward because of the increasing propensities of the reactions consuming $Y_i$'s.}

For our asymptotic time analysis, let the input size $n =  \|\vx\|$ be the number of input molecules.
In this section, the total molecular count attainable will always be $O(n)$;
thus, by finite density constraint, the volume $v = O(n)$.

We require the following theorem, due to Angluin, Aspnes, Diamadi, Fischer, and Ren\'{e}~\cite{angluin2006passivelymobile}, which states that any semilinear predicate can be decided by a CRD in expected time $O(n \log n)$.  (This was subsequently reduced to $O(n)$ by Angluin, Aspnes, and Eisenstat~\cite{angluin2006fast}, but $O(n \log n)$ suffices for our purpose.)

\begin{theorem}[\cite{angluin2006passivelymobile}]
  Let $\phi:\N^k\to\{0,1\}$ be a semilinear predicate.
  Then there is a stable CRD $\calD$ that decides $\phi$, and the expected time to reach an output-stable state on input is $O(n \log n)$.
\end{theorem}

Throughout this section, we use the technique of ``running multiple CRNs in parallel'' on the same input.
To accomplish this it is necessary to split the inputs $X_1,\ldots,X_k$ into separate molecules using a reaction $X_i \to X_i^1 + X_i^2 + \ldots + X_i^p$, which will add only $O(\log n)$ to the time complexity, so that each of the $p$ separate parallel CRNs do not interfere with one another.
For brevity we omit stating this formally when the technique is used.

The next lemma shows that affine partial functions can be computed in expected time $O(n \log n)$ by a CRC.
For its use in proving Theorem~\ref{lem-compute-semilinear-n-log-n}, we require that the output molecules be produced monotonically.
This is impossible for general affine partial functions.
For example, consider the function $f(x_1,x_2) = x_1 - x_2$ where $\dom f = \setl{(x_1,x_2)}{x_1 \geq x_2}$.
By withholding a single copy of $X_2$ and letting the CRC stabilize to the output value $\# Y = x_1-x_2+1$, then allowing the extra copy of $X_2$ to interact, the only way to stabilize to the correct output value $x_1-x_2$ is to consume a copy of the output species $Y$.
Therefore Lemma~\ref{lem-compute-linear-fast} is stated in terms of an encoding of affine partial functions that allows monotonic production of outputs, encoding the output value $\vy(j)$ as the difference between the counts of two monotonically produced species $Y_j^P$ and $Y_j^C$, using the same technique used in the proofs of Lemmas~\ref{lem-comp-pres} and~\ref{lem-pres-comp}.

Let $f:\N^k\dashrightarrow\N^l$ be an affine partial function, where, letting $\vy = f(\vx)$, for all $j \in \{1,\ldots,l\}$,  $\vy(j) = b_j + \frac{1}{d_j} \sum_{i=1}^k n_{i,j} (\vx(i) - c_i)$ for integer $n_{i,j}$ and nonnegative integer $b_j$, $c_i$, and $d_j$.
Define $\hat{f}:\N^k\dashrightarrow\N^{l}\times\N^{l}$ as follows.
For each $\vx\in \dom f$, define $\vy_C\in\N^l$ for each $j \in\{1,\ldots,l\}$ as $\vy_C(j) = - \frac{1}{d_j} \sum_{i=1}^k \min\{0,n_{i,j}\} (\vx(i) - c_i)$.
That is, $\vy_C(j)$ is the negation of the $j$'th coordinate of the output if taking the weighted sum of the inputs on only those coordinates with a negative coefficient $n_{i,j}$.
The value $\vy_P(j)$ is then similarly defined for all the positive coefficients and the $b_j$ offset: for each $\vx\in \dom f$, define $\vy_P\in\N^l$ for each $j \in\{1,\ldots,l\}$ as $\vy_P(j) = b_j + \frac{1}{d_j} \sum_{i=1}^k \max\{0,n_{i,j}\} (\vx(i) - c_i)$.
Because $\vx(i)-c_i \geq 0$, $\vy_P$ and $\vy_C$ are always nonnegative.
Then if $\vy = f(\vx)$, we have that $\vy = \vy_P - \vy_C$.
Define $\hat{f}$ as $\hat{f}(\vx) = (\vy_P,\vy_C)$.

\begin{lem} \label{lem-compute-linear-fast}
  Let $f:\N^k\dashrightarrow\N^l$ be an affine partial function.
  Then there is a CRC that computes $\hat{f}:\N^k\dashrightarrow\N^{l}\times\N^{l}$ in expected time $O(n \log n)$, such that the output molecules monotonically increase with time (i.e. none are ever consumed), and at most $O(n)$ molecules are ever produced.
\end{lem}

\newcommand{\hL}{\widehat{L}}

\begin{proof}
  If $(\vy_P,\vy_C) = \hat{f}(\vx)$, then there exist $kl$ integers $n_{1,1},\ldots,n_{k,l}\in\Z$ and $2l+k$ nonnegative integers $b_1,\ldots,b_l,c_1,\ldots,c_k,d_1,\ldots,d_l\in\N$ such that, if $\vy = f(\vx)$, then for each $j \in \{1,\ldots,l\}$,
$\vy_C(j) = - \sum_{i=1}^k \frac{1}{d_j} \min\{0,n_{i,j}\} (\vx(i) - c_i)$ and $\vy_P(j) = b_j + \frac{1}{d_j} \sum_{i=1}^k \max\{0,n_{i,j}\} (\vx(i) - c_i)$.
  Define the CRC as follows.
  It has input species $\Sigma = \{X_1,\ldots,X_k\}$ and output species $\Gamma = \{Y_1^P,\ldots,Y_l^P,Y_1^C,\ldots,Y_l^C\}$.

  For each $j \in \{1,\ldots,l\}$, start with $b_j$ copies of $Y_j^P$.
  This accounts for the $b_j$ offsets.

  For each $i \in \{1,\ldots,k\}$, start with a single molecule $C_{i}^{0}$, and for each $m \in \{0,\ldots,c_{i}-1\}$, add the reactions
\begin{eqnarray}
    C_{i}^{m} + X_{i} &\to& C_{i}^{m+1} \label{rxn-leader-count-c-offset}
\\
    C_{i}^{c_{i}} + X_{i} &\to& C_{i}^{c_{i}} + X'_{i} \label{rxn-leader-convert-post-offset-X}
\end{eqnarray}
This accounts for the $c_i$ offsets by eventually producing $\vx(i) - c_i$ copies of $X'_i$.
Reaction~\eqref{rxn-leader-count-c-offset} takes time $O(n)$ to complete because each reaction takes time at most $O(n)$ and a constant number, $c_i$, of such reactions must take place.
Once $C_{i}^{c_{i}}$ is produced (hence there are now $\vx(i) - c_i$ copies of $X_i$), reaction~\eqref{rxn-leader-convert-post-offset-X} takes time $O(n \log n)$ to complete by a coupon collector argument.

For each $i \in \{1,\ldots,k\}$, add the reaction
  \begin{eqnarray}
    X'_i &\to& X_{i,1} + X_{i,2} + \ldots + X_{i,l} \label{rxn-X-duplicate-for-each-output}
  \end{eqnarray}
This allows each output to be associated with its own copy of the input.
Reaction~\eqref{rxn-X-duplicate-for-each-output} takes time $O(\log n)$ to complete.

For each $i \in \{1,\ldots,k\}$ and $j \in \{1,\ldots,l\}$, if $n_{i,j} > 0$, add the reaction
\begin{eqnarray}
    X_{i,j} &\to& n_{i,j} Z_{j}^P \label{rxn-X-convert-to-Z-positive}
\end{eqnarray}
and if $n_{i,j} < 0$, add the reaction
\begin{eqnarray}
    X_{i,j} &\to& (-n_{i,j}) Z_{j}^C \label{rxn-X-convert-to-Z-negative}
\end{eqnarray}
Reaction \eqref{rxn-X-convert-to-Z-positive} produces $d_j (\vy_P(j) - b_j)$ copies of $Z_j^P$, and reaction \eqref{rxn-X-convert-to-Z-negative} produces $d_j \vy_C(j)$ copies of $Z_j^C$.
Each takes time $O(\log n)$ to complete.

Finally, to produce the correct number of $Y_j^P$ and $Y_j^C$ output molecules, we must divide the count of each $Z_j^P$ and $Z_j^C$ by $d_j$.
For each $j \in \{1,\ldots,l\}$, start with a single copy of a molecule $D_{j}^{0,P}$ and another $D_{j}^{0,C}$.
For each $j \in \{1,\ldots,l\}$ and each $m\in\{0,\ldots,d_{j}-1\}$, add the reactions
\begin{eqnarray*}
    D_{j}^{m,P} + Z_{j}^P &\to& \left\{
                              \begin{array}{ll}
                                D_{j}^{m+1,P}, & \hbox{if $m < d_{j}-1$;} \\
                                D_{j}^{0,P} + Y_j^P, & \hbox{if $m = d_{j}-1$.}
                              \end{array}
                              \right.
\\
    D_{j}^{m,C} + Z_{j}^C &\to& \left\{
                              \begin{array}{ll}
                                D_{j}^{m+1,C}, & \hbox{if $m < d_{j}-1$;} \\
                                D_{j}^{0,C} + Y_j^C, & \hbox{if $m = d_{j}-1$.}
                              \end{array}
                              \right.
\end{eqnarray*}
These reactions implement this division.
By a coupon collector argument, they each require time $O(n \log n)$ to complete.
\qedl\end{proof}

The next lemma shows that every semilinear function $f$ can be computed by a CRC in $O(n \log n)$ time.
It uses a systematic construction based on breaking down $f$ into a finite number of partial affine functions
$f_1,\ldots,f_m$, in which deciding which $f_i$ to apply is itself a semilinear predicate.
Intuitively, the construction proceeds by running many CRCs and CRDs in parallel on input $\vx$, computing all $f_i$'s and all predicates of the form $\phi_i =$ ``$\vx \in \dom f_i$?''
The $\phi_i$ predicate computation is used to activate (in the case of a ``yes'' answer) or deactivate (in case of ``no'') the outputs of $f_i$.
Since eventually one CRD stabilizes to ``yes'' and the remainder to ``no'', eventually the outputs of one $f_i$ are activated and the remainder deactivated, so that the value $f(\vx)$ is properly computed.

\begin{lem} \label{lem-compute-semilinear-n-log-n}
  Let $f:\N^k\to\N^l$ be semilinear.  Then there is a CRC $\calC$ that stably computes $f$, and the expected time for $\calC$ to reach a count-stable configuration on input $\vx$ is $O(n \log n)$ (where the $O()$ constant depends on $f$).
\end{lem}


\begin{proof}  
  By Lemma~\ref{lem-pres-comp}, there is a CRC $\calC_s$ that stably computes $f$.
  However, that CRC is too slow to use in this proof.
  We provide an alternative proof that every semilinear function can be computed by a CRC in expected time $O(n \log n)$.
  Rather than relying on a random search of the output space as in Lemma~\ref{lem-pres-comp}, it computes the function more directly.
Our CRC will have input species $\Sigma = \{X_1, \ldots, X_k\}$ and output species $\Gamma = \{Y_1,\ldots,Y_l\}$.

  By Lemma~\ref{lem-semilinear-function-finite-union-linear}, there is a finite set $F=\{f_1:\N^k \dashrightarrow \N^l,\ldots,f_m:\N^k \dashrightarrow \N^l\}$ of affine partial functions, where each $\dom f_i$ is a linear set, such that, for each $\vx\in\N^k$, if $f_i(\vx)$ is defined, then $f(\vx)=f_i(\vx)$.
  We compute $f$ on input $\vx$ as follows.
  Since each $\dom f_i$ is a linear (and therefore semilinear) set, we compute each predicate $\phi_i = $ ``$\vx \in \dom f_i$ and $(\forall i' \in \{1,\ldots,i-1\})\ \vx \not\in \dom f_{i'}$?'' by separate parallel CRD's.
  (The latter condition ensures that for each $\vx$, precisely one of the predicates is true, in case the domains of the partial functions have nonempty intersection.)

  By Lemma~\ref{lem-compute-linear-fast}, we can compute each $\hat{f}_i$ by parallel CRC's.
  Assume that for each $i \in \{1,\ldots,m\}$ and each $j \in \{1,\ldots,l\}$, the $j$th pair of outputs $\vy_P(j)$ and $\vy_C(j)$ of the $i$th function is represented by species $\hY_{i,j}^P$ and $\hY_{i,j}^C$.
  We interpret each $\hY_{i,j}^P$ and $\hY_{i,j}^C$ as an ``inactive'' version of ``active'' output species $Y_{i,j}^P$ and $Y_{i,j}^C$.

  For each $i \in \{1,\ldots,m\}$, we assume that the CRD computing the predicate $\phi_i$ represents its output by voting species $L_i^1$ to represent ``yes'' and $L_i^0$ to represent ``no''.
  Then add the following reactions for each $i \in \{1,\ldots,m\}$ and each $j \in \{1,\ldots,l\}$:
  \begin{eqnarray*}
    L_i^1 + \hY_{i,j}^P &\to& L_i^1 + Y_{i,j}^P + Y_j
    \\
    L_i^0 + Y_{i,j}^P &\to& L_i^0 + M_{i,j}
    \\
    M_{i,j} + Y_j &\to& \hY_{i,j}^P
  \end{eqnarray*}
(The latter two reactions implement the reverse direction of the first reaction using only bimolecular reactions.)
Also add the reactions
  \begin{eqnarray*}
    L_i^1 + \hY_{i,j}^C &\to& L_i^1 + Y_{i,j}^C
    \\
    L_i^0 + Y_{i,j}^C &\to& L_i^0 + \hY_{i,j}^C
  \end{eqnarray*}
and
  \begin{eqnarray*}
    Y_{i,j}^P + Y_{i,j}^C &\to& K_{j}
    \\
    K_{j} + Y_j &\to& \emptyset
  \end{eqnarray*}

  That is, a ``yes'' answer for function $i$ activates the $i$th output and a ``no'' answer deactivates the $i$th output.
  Eventually each CRD stabilizes so that precisely one $i$ has $L_i^1$ present, and for all $i' \neq i$, $L_{i'}^0$ is present.
  At this point, all outputs for the correct function $\hat{f}_i$ are activated and all other outputs are deactivated.
  The reactions enforce that at any time, $\# Y_j = \sum_{i=1}^m \# Y_{i,j}^P + \# K_j + \# M_{i,j}$.
  In particular, $\# Y_j \geq \# K_j$ and $\# Y_j \geq \# M_{i,j}$ at all times, so there will never be a $K_j$ or $M_{i,j}$ molecule that cannot participate in the reaction of which it is a reactant.
  Eventually $\# Y_{i,j}^P$ and $\# Y_{i,j}^C$ stabilize to 0 to for all but one value of $i$ (by the fifth reaction), and for this value of $i$, $\# Y_{i,j}^P$ stabilizes to $\vy(j)$ and $\# Y_{i,j}^C$ stabilizes to 0 (by the second-to-last reaction).
  Eventually $\# K_j$ stabilizes to 0 by the last reaction.
  Eventually $\# M_{i,j}$ stabilizes to 0 since $L_i^0$ is absent for the correct function $\hat{f}_i$.
  This ensures that $\# Y_j$ stabilizes to $\vy(j)$.

  It remains to analyze the expected time to stabilization.
  Let $n = \|\vx\|$.
  By Lemma~\ref{lem-compute-linear-fast}, the expected time for each affine function computation to complete is $O(n \log n)$.
  Since the $\hY_{i,j}^P$ are produced monotonically, the most $Y_{i,j}^P$ molecules that are ever produced is $\#_\infty \hY_{i,j}^P$.
  Since 
we have $m$ computations in parallel, the expected time for all of them to complete is at most $O((n \log n) m) = O(n \log n)$ (since $m$ depends on $f$ but not $n$).
  We must also wait for each predicate computation to complete.
  By Theorem~\ref{thm-semilinear}, each of these predicates takes expected time at most $O(n \log n)$ to complete, so all of them complete in expected time at most $O(m n \log n) = O(n \log n)$.

  At this point, the $L^i_1$ leaders must convert inactive output species to active, and $L^{i'}_0$ (for $i' \neq i$) must convert active output species to inactive.
  A similar analysis to the proof of Lemma~\ref{lem-compute-semilinear-n-log-n} shows that each of these requires at most $O(n \log n)$ expected time, therefore they all complete in expected time at most $O((n \log n) m) = O(n \log n)$.
  Finally, a similar argument shows that it requires at most expected time $O(n \log n)$ for the final two reactions to consume all $Y_{i,j}^C$ and $K_j$ molecules, at which point the system has stabilized.
\qedl\end{proof}

\begin{lem}\label{lem-semilinear-function-finite-union-linear}
  Let $f:\N^k\to\N^l$ be a semilinear function.
  Then there is a finite set $\{f_1:\N^k \dashrightarrow \N^l,\ldots,f_m:\N^k \dashrightarrow \N^l\}$ of affine partial functions, where each $\dom f_i$ is a linear set, such that, for each $\vx\in\N^k$, if $f_i(\vx)$ is defined, then $f(\vx)=f_i(\vx)$, and $\bigcup_{i=1}^m \dom f_i = \N^k$.
\end{lem}

We split the semilinear function into partial functions, each with a graph that is a linear set.
The non-trivial aspect of our argument is showing that (straightforward) linear algebra can be used to solve our problem about integer arithmetic.
For example, consider a partial function defined by the following linear graph:
$\vec{b} = \vec{0}$, $\vec{u}_1 = (1, 1, 1)$, $\vec{u}_2 = (2, 0, 1)$, $\vec{u}_3 = (0, 2, 1)$ (where the first two coordinates are inputs and the last coordinate is the output).
Note that the set of points where this function is defined is where $x_1 + x_2$ is even.
Given an input point $\vec{x}$,
the natural approach to evaluating the function
is to solve for the coefficients $n_1, n_2, n_3$ such that
$\vec{x}$ can be expressed as a linear combination of $\vec{u}_1, \vec{u}_2, \vec{u}_3$ restricted to the first two coordinates.
Then the linear combination of the last coordinate of $\vec{u}_1, \vec{u}_2, \vec{u}_3$ with coefficients $n_1, n_2, n_3$ would give the output.
However, the vectors $\vec{u}_1, \vec{u}_2, \vec{u}_3$ are not linearly independent (yet this linear set cannot be expressed with less than three basis vectors --- illustrating the difference between real spaces and integer-valued linear sets), so there are infinitely many real-valued solutions for the coefficients.
We show that $\vec{u}_i$ must span a real subspace with at most one output value for any input coordinates.
Then we can throw out a vector (say $\vec{u}_1$) to obtain a set of linearly independent vectors ($\vec{u}_2, \vec{u}_3$) and solve for $n_2, n_3 \in \R$, and let $n_1 = 0$.
In this example, the resulting partial affine function is $f(x_1,x_2) = (x_1 + x_2)/2$.

\begin{proof}
  Let
  $
    F = \setl{ (\vx,\vy) \in \N^k \times \N^l }{ f(\vx) = \vy }
  $
  be the graph of $f$.
  Since $F$ is semilinear, it is a finite union of linear sets $\{L_1,\ldots,L_n\}$.
  It suffices to show that each of these linear sets $L_m$ is the graph of an affine partial function.
  Since $L_m$ is linear, its projection onto any subset of its coordinates is linear.
  Therefore $\dom f_m$ (the projection of $L_m$ onto its first $k$ coordinates) is linear.

  We consider each output coordinate separately, since if we can show that each $\vy(j)$ is an affine function of $\vx$, then it follows that $\vy$ is an affine function of $\vx$.
  Fix $j \in \{1,\ldots,l\}$. Let $L'_m$ be the $(k+1)$-dimensional projection of $L_m$ onto the coordinates defined by $\vx$ and $\vy(j)$, which is linear because $L_m$ is.
  Since $L'_m$ is linear, there exist vectors $\vec{b},\vec{u}_1,\ldots,\vec{u}_p \in\N^{k+1}$ such that $L'_m=\setl{\vec{b} + n_1 \vec{u}_1 + \ldots + n_p \vec{u}_p}{n_1,\ldots,n_p\in\N}.$


Consider the real-vector subspace spanned by $\vec{u}_1, \ldots, \vec{u}_p$.
It cannot contain the vector $\vec{j} = (0,\ldots,0,1)^T$.
Suppose it does.
Take a subset of linearly independent vectors spanning this subspace from the above list (we possibly remove some linearly dependent vectors);
say $\vec{u}_1, \ldots, \vec{u}_{p'}$.
The unique solution to the coefficients $\xi_1, \ldots, \xi_{p'} \in \R$ such that $\vec{j} = \xi_1 \vec{u}_1 + \ldots + \xi_{p'} \vec{u}_{p'}$ can be obtained by using the left-inverse of the matrix with columns $\vec{u}_1, \ldots, \vec{u}_{p'}$ (the left inverse exists because the matrix is full-rank).
Since the elements of the left-inverse matrix are rational functions of the matrix elements,
and vectors $\vec{u}_1, \ldots, \vec{u}_{p'}$ consist of numbers in $\N$, the coefficients $\xi_1, \ldots, \xi_{p'}$ are rational.
We can multiply all the coefficients by the least common multiple of their denominators $c$ yielding
$c \vec{j} = m_1 \vec{u}_1 + \ldots + m_{p'} \vec{u}_{p'}$ where $m_1, \ldots, m_{p'} \in \Z$.
Now consider a point $\vec{a}$ in $L'_m$ defined as $\vec{b} + n_1 \vec{u_1} + \ldots + n_{p'} \vec{u}_{p'}$, where $n_i \in \N$.
We choose $\vec{a}$ such that $n_i$ are large enough that $n'_i \triangleq n_i + m_i \geq 0$.
Since $n'_i \in \N$, we have that both $\vec{a}$ and $\vec{a} + c \vec{j} = \vec{b} + n'_1 \vec{u_1} + \ldots + n'_{p'} \vec{u}_{p'}$ are in $L'_m$.
This is a contradiction because $L'_m$ is the graph of a partial function and cannot contain two different points that agree on their first $k$ coordinates.
Therefore $\vec{j}$ is not contained in the span of $\vu_1,\ldots,\vu_p$.

Consider again the real-vector subspace spanned by $\vec{u}_1, \ldots, \vec{u}_{p}$.
Again, let $\vec{u}_1, \ldots, \vec{u}_{p'}$ be a subset of linearly independent vectors spanning this subspace.
Since $\vec{j}$ is not in it, the subspace must be at most $k$ dimensional.
If it is strictly less than $k$ dimensional, add enough vectors in $\N^{k+1}$ to the basis set for the spanned subspace to be exactly $k$-dimensional but not include $\vec{j}$.
Call this new set of $k$ linearly independent vectors $\vec{w}_1, \ldots, \vec{w}_{k}$, where $\vec{w}_i = \vec{u}_i$ for $i \in \{1,\ldots,p'\}$.
Let $\vec{v}_1, \ldots, \vec{v}_{k} \in \N^{k}$ be $\vec{w}_1, \ldots, \vec{w}_{k}$ restricted to the first $k$ coordinates.
The fact that $\vec{w}_1, \ldots, \vec{w}_{k}$ are linearly independent, but $\vec{j}$ is not in the subspace spanned by them, implies that $\vec{v}_1, \ldots, \vec{v}_k$ are linearly independent as well.
This can be seen as follows.
If $\vec{v}_1, \ldots, \vec{v}_k$ were not linearly independent,
then we could write $\vec{v}_k = \xi_1 \vec{v}_1 + \ldots + \xi_{k-1} \vec{v}_{k-1}$ for some $\xi_i \in \R$.
However, $\vec{w}_k \neq \vec{w}'_k \triangleq \xi_1 \vec{w}_1 + \ldots + \xi_{k-1} \vec{w}_{k-1}$.
Since $\vec{j}$ is proportional to $\vec{w}'_k - \vec{w}_k$, we obtain a contradiction.
Therefore $\vec{v}_1, \ldots, \vec{v}_k$ are linearly independent.

We now describe how to construct an affine function $\vy(j) = f(\vx)$ for $L'_m$ from $\vec{w}_1, \ldots, \vec{w}_{k}$.
Let matrix $\vec{V}$ be the square matrix with $\vec{v}_1, \ldots, \vec{v}_{k}$ as columns.
Let $\vb'$ be $\vb$ restricted to its first $k$ coordinates.
We claim that $\vy(j)  = \vb(k+1) + (\vec{w}_1(k+1), \ldots, \vec{w}_{k}(k+1)) \cdot \vec{V}^{-1} \cdot \left( \vx - \vb' \right)$.
Below we'll show that this expression computes the correct value $\vy(j)$.
But first we show that it defines a partial affine function $f(\vx)$.
Because $\vec{v}_1, \ldots, \vec{v}_{k}$ are linearly independent, the inverse $\vec{V}^{-1}$ is well-defined.
We need to show $f(\vx) = b_j + \frac{1}{d_j} \sum_{i=1}^k n_{i,j} (\vx(i) - c_i)$ for integer $n_{i,j}$ and nonnegative integer $b_j$, $c_i$, and $d_j$,
and that on the domain of $f$, $\vx(i) - c_i \geq 0$.
The offset $b_j = \vb(k+1)$, which is a non-negative integer because $\vb$ is a vector of non-negative integers.
Since the offset vector $\vb'$ is the same for each output dimension, and it is likewise non-negative, we obtain the offset $c_i = \vb'(i)$.
Further, since $\vec{V^{-1}}$ consists of rational elements (because $\vec{V}$ consists of elements in $\N$), we can define $d_j$ and $n_{i,j}$ as needed.
Finally, note that the least value of $\vx(i)$ that could be in $L'_m$ is $\vb'(i) = c_i$,
and thus on the domain of $f$, $\vx(i) - c_i \geq 0$.

Finally, we show that this expression computes the correct value $\vy(j)$.
Let $(\xi_1, \dots, \xi_k)^T \triangleq \vec{V}^{-1} \cdot (\vx - \vb')$,
which implies that $\vx  = \vb' + \sum_{i = 1}^k \xi_i \vv_i$.
If our value of $\vy(j)$ is incorrect, then $\exists n_1, \dots, n_p \in \N$ such that
$\vb + \sum_{i=1}^p n_i \vu_i$ and $\vb + \sum_{i=1}^k \xi_i \vw_i$ agree on the first $k$ coordinates but not on the $k+1$st.
Recall that the real-vector subspace spanned by $\vw_1, \dots, \vw_k$ includes the subspace spanned by $\vu_1, \dots, \vu_p$ but does not include $\vec{j}$.
But $\sum_{i=1}^p n_i \vu_i - \sum_{i=1}^k \xi_i \vw_i$ is proportional to $\vec{j}$
and lies in the subspace spanned by $\vw_1, \dots, \vw_k$.
Therefore we obtain a contradiction, implying that our value of $\vy(j)$ is computed correctly.
\qedl\end{proof}

Angluin, Aspnes, and Eisenstat combined the slow, deterministic predicate-deciding results of \cite{AngluinAE06} with a fast, error-prone simulation of a bounded-space Turing machine to show that semilinear predicates can be computed without error in expected polylogarithmic time~\cite{angluin2006fast}.
We show that a similar technique implies that semilinear functions can be computed by CRNs without error in expected polylogarithmic time in the kinetic model, combining the same Turing machine simulation with our $O(n \log n)$ construction described in Lemma~\ref{lem-compute-semilinear-n-log-n}.

We in fact use the same construction of Angluin, Aspnes, and Eisenstat~\cite{angluin2006fast} in order to conduct the fast, error-prone computation in our proof of Theorem~\ref{thm-func-speed}.
The next theorem formalizes the properties of their construction that we require.

\begin{theorem}[\cite{angluin2006fast}] \label{thm-func-error}
  Let $f:\N^k\to\N^l$ be a function by a $t(m)$-time-bounded, $s(m)$-space-bounded Turing machine, where $m \approx \log n$ is the input length in binary, and let $c\in\N$.
  Then there is a CRC $\calC$ that computes $f$ correctly with probability at least $1- n^{-c}$, and the expected time for $\calC$ to reach a count-stable configuration is $O(t(m)^5)$.
  Furthermore, the total molecular count never exceeds $O(2^{s(m)})$.
\end{theorem}

Semilinear functions on an $m$-bit input can be computed in time $O(m)$ and space $O(m)$ on a Turing machine.
Therefore the bounds on CRC expected time and molecular count stated in Theorem~\ref{thm-func-error} are $O(\log^5 n)$ and $O(n)$, respectively, when expressed in terms of the number of input molecules $n$.

Although Angluin, Aspnes, and Eisenstat~\cite{angluin2006fast} exclusively use two-reactant, two-product reactions, and not all of the properties stated in Theorem~\ref{thm-func-error} are explicitly stated in~\cite{angluin2006fast}, their construction can be easily modified to have the stated properties.
Since that construction preserves the total molecular count, they require some non-uniformity to supply enough ``fuel'' molecules $F$, based on the space usage $s(m)$ (which varies with the input size), so that the tape of the Turing machine can be accurately represented throughout the computation.
However, in our model, molecules may be produced.
We compute semilinear functions, where as observed above has total molecular count bounded by $O(n)$, so these fuels may be supplied by letting the first reaction of the input $X_i$ be $X_i \to X_i' + c F$, where $X_i'$ is the input interacting with the rest of the CRC, and $c \in \N$ is chosen sufficiently large.

The following theorem is the main theorem of this section.

\begin{theorem} \label{thm-func-speed}
  Let $f:\N^k\to\N^l$ be semilinear.  Then there is a CRC $\calC$ that stably computes $f$, and the expected time for $\calC$ to reach a count-stable configuration is $O(\log^5 n)$.
\end{theorem}

\begin{proof}
  Our CRC will use the counts of $Y_j$ for each output dimension $\vy(j)$ as the global output, and begins by running in parallel:
  \begin{enumerate}
  \item
    A fast, error-prone CRC $\mathcal{F}$  for $\vec{y},\vec{b},\vec{c}=f(\vx)$, as in Theorem~\ref{thm-func-error}. For any constant $c > 0$, we may design $\mathcal{F}$ so that it is correct and finishes in time $O(\log^5 n)$ with probability at least $1- n^{-c}$, while reaching total molecular count never higher than $O(n)$. We modify $\mathcal{F}$ so that upon halting, it copies an ``internal'' output species $\hY_j$ to $Y_j$ (the global output), $B_j$, and $C_j$ through reactions $H + \hY_j \to Y_j + B_j + C_j$ (in asymptotically negligible time).
    Here, $H$ is some molecule that is guaranteed with high probability not to be present until $\mathcal{F}$ has halted, and to be present in large ($\Omega(n)$) count so that the conversion is fast.
    In this way we are guaranteed that the amount of $Y_j$ produced by $\mathcal{C}$ is the same as the amounts of $B_j$ and $C_j$ no matter whether its computation is correct or not.

 \item
    A slow, deterministic CRC $\mathcal{S}$ for $\vy' = f(\vx)$. It is constructed as in Lemma~\ref{lem-compute-semilinear-n-log-n}, running in expected $O(n \log{n})$ time.

 \item
    A slow, deterministic CRD $\mathcal{D}$ for the semilinear predicate ``$\vb = f(\vx)$?". It is constructed as in Theorem~\ref{thm-semilinear} and runs in expected $O(n)$ time.
 \end{enumerate}

  Following Angluin, Aspnes, and Eisenstat~\cite{angluin2006fast}, we construct a ``timed trigger'' as follows, using a leader molecule, a marker molecule, and $n = \|\vx\|$ interfering molecules.
  These interfering molecules can simply be the input species and some of their ``descendants'' such that their count is held constant.
  This can be done for input species $X_i$ by a reaction such as $X_i \to I + X_i'$, where $X_i$ is (one of) the original input species, $I$ is the interfering molecule, and $X_i'$ is the input species interacting with the remainder of the CRC.
  The leader will then interact with both $X_i$ and $I$ as interfering molecules.

  The leader fires the trigger if it encounters the marker molecule $d$ times without any intervening reactions with the interfering molecules.
  This happens rarely enough that with high probability the trigger fires after $\mathcal{F}$ and $\mathcal{D}$ finishes (time analysis is presented below).
  When the trigger fires, it checks if $\mathcal{D}$ is outputting a ``no" (e.g.\ has a molecule of $L_0$), and if so, produces a molecule of $\Fix$.
  This indicates that the output of the fast CRC $\mathcal{F}$ is not to be trusted, and the system should switch from the possible erroneous result of $\mathcal{F}$ to the sure-to-be correct result of $\mathcal{S}$.

    Once a $\Fix$ is produced, the system converts the output molecules $Y'_j$ of the slow, deterministic CRC $\mathcal{S}$ to the global output $Y_j$, and kills enough of the global output molecules to remove the ones produced by the fast, error-prone CRC:

  \begin{eqnarray}
    \Fix + Y'_j   &\to&   \Fix + Y_j   \\
    \Fix + C_j    &\to&  \Fix + \overline{Y}_j     \\
    Y_j + \overline{Y}_j  &\to&  \emptyset.
  \end{eqnarray}

Finally, $\Fix$ triggers a process consuming all species of $\mathcal{F}$ other than $Y_j, B_j$, and $C_j$ in expected $O(\log n)$ time so that afterward, $\mathcal{F}$ cannot produce any output molecules.
More formally, let $Q_\mathcal{F}$ be the set of all species used by $\mathcal{F}$.
For all $X \in Q_\mathcal{F} \setminus \bigcup_{j=1}^l \{Y_j,B_j,C_j\}$, add the reactions
  \begin{eqnarray}
    \Fix + X   &\to&   \Fix + K   \\
    K + X    &\to&  K + K,
  \end{eqnarray}
where $K \not\in Q_\mathcal{F}$ is a unique species.

First, observe that the output will always eventually converge to the right answer, no matter what happens:
If $\Fix$ is eventually produced, then the output will eventually be exactly that given by $\mathcal{S}$ which is guaranteed to converge correctly.
If $\Fix$ is never produced, then the fast, error-prone CRC must produce the correct amount of $Y_j$ --- otherwise, $\mathcal{D}$ will detect a problem.

For the expected time analysis, let us first analyze the trigger.
The probability that the trigger leader will fire on any particular reaction number is at most $n^{-d}$.
In time $n^2$, the expected number of leader reactions is $O(n^2)$.
Thus, the expected number of firings of the trigger in $n^2$ time is $n^{-d+2}$.
This implies that the probability that the trigger fires before $n^2$ time is at most $n^{-d+2}$.
The expected time for the trigger to fire is $O(n^d)$.

We now consider the contribution to the total expected time from $3$ cases:
\begin{enumerate}
	\item $\mathcal{F}$ is correct, and the trigger fires after time $n^2$. There are two subcases:
	(a) $\mathcal{F}$ finishes before the trigger fires. Conditional on this, the whole system converges to the correct answer, never to change it again, in expected time $O(\log^5 n)$. This subcase contributes at most $O(\log^5 n)$ to the total expected time.
	(b) $\mathcal{F}$ finishes after the trigger fires. In this case, we may produce a $\Fix$ molecule and have to rely on the slow CRC $\mathcal{S}$. The probability of this case happening is at most $n^{-c}$. Conditional on this case, the expected time for the trigger to fire is still $O(n^d)$. The whole system converges to the correct answer in expected time $O(n^d)$, because everything else is asymptotically negligible. Thus the contribution of this subcase to the total expectation is at most $O(n^{-c} \cdot n^d) = O(n^{-c+d})$.
	\item $\mathcal{F}$ is correct, but the trigger fires before $n^2$ time. In this case, we may produce a $\Fix$ molecule and have to rely on the slow CRC $\mathcal{S}$ for the output. The probability of this case occurring is at most $n^{-d+2}$. Conditional on this case occurring, the expected time for the whole system to converge to the correct answer can be bounded by $O(n^2)$. Thus the contribution of this subcase to the total expectation is at most $O(n^{-d+2} \cdot n^2) = O(n^{-d+4})$.
	\item $\mathcal{F}$ fails. In this case we'll have to rely on the slow CRC $\mathcal{S}$ for the output again. Since this occurs with probability at most $n^{-c}$, and the conditional expected time for the whole system to converge to the correct answer can be bounded by $O(n^d)$ again, the contribution of this subcase to the total expectation is at most $O(n^{-c} \cdot n^d) = O(n^{-c+d})$.
\end{enumerate}

So the total expected time is bounded by $O(\log^5 n) +  O(n^{-c+d}) + O(n^{-d+4}) + O(n^{-c+d}) =  O(\log^5 n)$ for $d>4, c>d$.
\qedl\end{proof}


\section{Conclusion}
\label{sec-conclusion}

We defined deterministic computation of CRNs corresponding to the intuitive notion that certain systems are guaranteed to converge to the correct answer no matter what order the reactions happen to occur in.
We showed that this kind of computation corresponds exactly to the class of functions with semilinear graphs.
We further showed that all functions in this class can be computed efficiently.

A work on chemical computation can stumble by attempting to shoehorn an ill-fitting computational paradigm into chemistry.
While our systematic construction may seem complex,
we are inspired by examples like those shown in Fig.~\ref{fig:examples} that appear to be good fits to the computational substrate.
While delineation of computation that is ``natural'' for a chemical system is necessarily imprecise and speculative, it is examples such as these that makes us satisfied that we are studying a form of natural chemical computation.

Our systematic construction (unlike the examples in Fig.~\ref{fig:examples}) relies on a carefully chosen initial context --- the ``extra'' molecules that are necessary for the computation to proceed.
Some of these species need to be present in a single copy (``leader'').
We left unanswered whether it may be possible to dispense with this level of control of the chemical environment.
We suspect this generalization would be non-trivial because the problem of generating a prescribed molecular count of a species from an uncontrolled context is computationally challenging (see e.g.~the ``leader election'' problem~\cite{aspnes2007introduction}).

In contrast to the CRN model discussed in this paper, which is appropriate for small chemical systems in which every single molecule matters,
classical ``Avogadro-scale'' chemistry is modeled using real-valued concentrations that evolve according to mass-action ODEs.
Moreover, despite relatively small molecular counts, many biological chemical systems are well-modeled by mass-action ODEs.
While the scaling of stochastic CRNs to mass-action systems is understood from a dynamical systems perspective~\cite{kurtz1972relationship},
little work has been done comparing their computational abilities.
There are hints that single/few-molecule CRNs perform a fundamentally different kind of computation.
For example, recent theoretical work has investigated whether CRNs can tolerate multiple copies of the network running in parallel finding that they can lose their computational abilities~\cite{CondonHMT12,condon2012reachability}.

Does our notion of deterministic computation have an equivalent in mass-action systems?
Consider what happens when the CRN shown in Fig.~\ref{fig:examples}(c) is considered as a mass-action reaction network, with (non-negative) real-valued inputs $[X_1]_0$, $[X_2]_0$ and output $[Y]_\infty$ (where we use the standard mass-action convention: $[\cdot]_0$ for the initial concentration, and $[\cdot]_\infty$ for the equilibrium concentration).
In the limit $t \rightarrow \infty$, the mass-action system will converge to the correct output amount of $[Y]_\infty = \max([X_1]_0, [X_2]_0)$, and moreover, output amount is independent of what (non-zero) rate constants are assigned to the reactions.
Thus one is tempted to connect the notion of deterministic computation studied here and the property of robustness to parameters of a mass-action system.
Parameter robustness is a recurring motif in biologically relevant reaction networks due to much evidence that biological systems tend to be robust to parameters~\cite{barkal1997robustness}.

However, the connection is not simple.
Consider the CRN shown in Fig.~\ref{fig:examples}(a).
In the mass-action limit it loses the ability of computing the floor function, but still computes $[Y]_\infty = [X]_0 / 2$ for real valued $[X]_0$, $[Y]_\infty$,
independent of reaction rates.
More interestingly, the CRN shown in Fig.~\ref{fig:examples}(b), when considered as mass-action reaction network, could converge to a different amount of $Y$ as $t \rightarrow \infty$, depending on the rate constants of the last two reactions and the input amounts.
Specifically, let $k_1$, $k_2$, and $k_3$ be the rate constants of the three reactions, respectively.
If $[X_1]_0 > [X_2]_0$ and $k_2 \leq k3 [X_2]_0 / ([X_1]_0 - [X_2]_0)$,
then $Y$ will go to $k_2/k_3 ([X_1]_0 - [X_2]_0)$ rather than $[X_2]_0$ as in Fig.~\ref{fig:examples}(b).
In all other cases, the output will correctly match the function in the figure.
(This can be verified by determining the steady states of the system and then determining the stability of each one as a function of the initial concentrations and rate constants.)
The cause of the disagreement between stochastic and mass-action instances of this CRN can be identified with the ``type I'' deviant effect demarcated by Samoilov and Arkin~\cite{samoilov2006deviant}.

\paragraph{Acknowledgements.}
We thank Damien Woods and Niranjan Srinivas for many useful discussions,
Monir Hajiaghayi for pointing out a problem in an earlier version of this paper,
and anonymous reviewers for helpful suggestions.

\bibliographystyle{plain}
\bibliography{tam}


\end{document}